\documentclass{article}

\usepackage{amsmath}
\usepackage{amsthm}
\usepackage{newtxtext}
\usepackage[subscriptcorrection]{newtxmath}
\usepackage{xcolor,colortbl}
\usepackage{comment}
\usepackage[font={small}]{caption}
\usepackage{subcaption}
\usepackage{placeins}
\usepackage{booktabs}
\usepackage{url}
\usepackage{xr}
\usepackage{natbib}
\usepackage{bbm}
\usepackage{graphicx}
\usepackage[affil-it]{authblk}
\usepackage[plain,noend]{algorithm2e}

\newcommand{\comments}[1]{}

\newcommand{\R}{\mathbb{R}}
\newcommand{\1}{\mathbbm{1}}

\newcommand{\A}{\mathcal{A}}
\newcommand{\X}{\mathcal{X}}
\newcommand{\Y}{\mathcal{Y}}

\newtheorem{theorem}{Theorem}  
\newtheorem{proposition}{Proposition}  
\newtheorem{corollary}{Corollary}  
\newtheorem{remark}{Remark} 

\providecommand{\keywords}[1]{\texttt{Keywords:} #1}
\newcommand{\email}[1]{\texttt{#1}}

\begin{document}
\markboth{J.Kampe et~al.}{Nested exemplar latent space models for dimension reduction
in dynamic networks}

\title{Nested exemplar latent space models for dimension reduction in dynamic networks}

\author[1]{J. N. KAMPE$^*$}
\author[2]{L. A. SILVA$^*$}
\author[3,4]{T. ROSLIN}
\author[1]{D. B. DUNSON}

\affil[1]{Department of Statistical Science, Duke University, Box 90251, Durham, North Carolina 27708, U.S.A. \protect\\ \email{jennifer.kampe@duke.edu}\protect\\
\email{dunson@duke.edu}}
\affil[2]{Department of Decision Sciences, Bocconi University, Via Röntgen 1, 20136 Milan, Italy \protect\\ \email{silva.luca@phd.unibocconi.it}}
\affil[3]{Department of Ecology, Swedish University of Agricultural Sciences, Ulls väg 18B, 75651 Uppsala, Sweden \protect\\ \email{tomas.roslin@slu.se}}
\affil[4]{Spatial Foodweb Ecology Group, Organismal and Evolutionary Biology Research Programme, Faculty of Biological and Environmental Sciences, University of Helsinki, PO Box 65 (Viikinkaari 1), FI-00014 University of Helsinki, Finland}

\maketitle
\def\thefootnote{*}\footnotetext{These authors contributed equally to this work.}\def\thefootnote{\arabic{footnote}}

\begin{abstract}
Dynamic latent space models are widely used for characterizing changes in networks and relational data over time. These models assign to each node latent attributes that characterize connectivity with other nodes, with these latent attributes dynamically changing over time.  Node attributes can be organized as a three-way tensor with modes corresponding to nodes, latent space dimension, and time. Unfortunately, as the number of nodes and time points increases, the number of elements of this tensor becomes enormous, leading to computational and statistical challenges, particularly when data are sparse. We propose a new approach for massively reducing dimensionality by expressing the latent node attribute tensor as low rank. This leads to an interesting new {\em nested exemplar} latent space model, which characterizes the node attribute tensor as dependent on low-dimensional exemplar traits for each node, weights for each latent space dimension, and exemplar curves characterizing time variation. We study properties of this framework, including expressivity, and develop efficient Bayesian inference algorithms. The approach leads to substantial advantages in simulations and applications to ecological networks.
\end{abstract}


\section{Introduction}

In fields ranging from ecology to finance, there is interest in studying how networks change over time, and in particular, how connectivity patterns vary dynamically. We focus on the scenario in which there is a common set of nodes across time, potentially subject to missingness, but the edges indicating interactions between nodes vary dynamically. As a motivating application, we analyze bipartite insect-plant pollination networks in Greenland's Zackenberg Valley. Situated in the high Arctic tundra, the site is disproportionately vulnerable to the effects of climate change (\cite{SCHMIDT20233244}; \cite{Rantenen2022}) and therefore of great interest is studying the impacts of long-term temperature change and associated phenological shifts on these essential pollination networks (\cite{Cirtwill}; \cite{Ascanio2024}). 

Bipartite interaction networks, in which interactions are recorded between species of two distinct types, emerge naturally throughout the ecological literature, as in the plant-pollinator networks of \cite{Cirtwill}, the primate-parasite networks of \cite{Herrera}, and the frugivore-tree networks of \cite{DurandBessart}. As a mathematical object, a binary bipartite dynamic network can be represented as a three-mode tensor $\mathcal{A}\in\R^{N\times M \times T}$, where $N$ and $M$ indicate the number of nodes of each type, $T$ indicates the number of discrete time periods, and each component $\mathcal{A}_{ijt}$ indicates the observation of an edge between nodes $i$ and $j$ at time $t$. Figure \ref{fig:1} presents an example of this data structure taken from the Zackenberg Valley data set, illustrating the evolution of a subnetwork consisting of the 15 most common plants and insects throughout a portion of the 2010 summer season. These networks are characterized by high sparsity, even for these comparatively dense subnetworks. 

\begin{figure}[h]
    \centering
   
    \includegraphics[width = \textwidth]{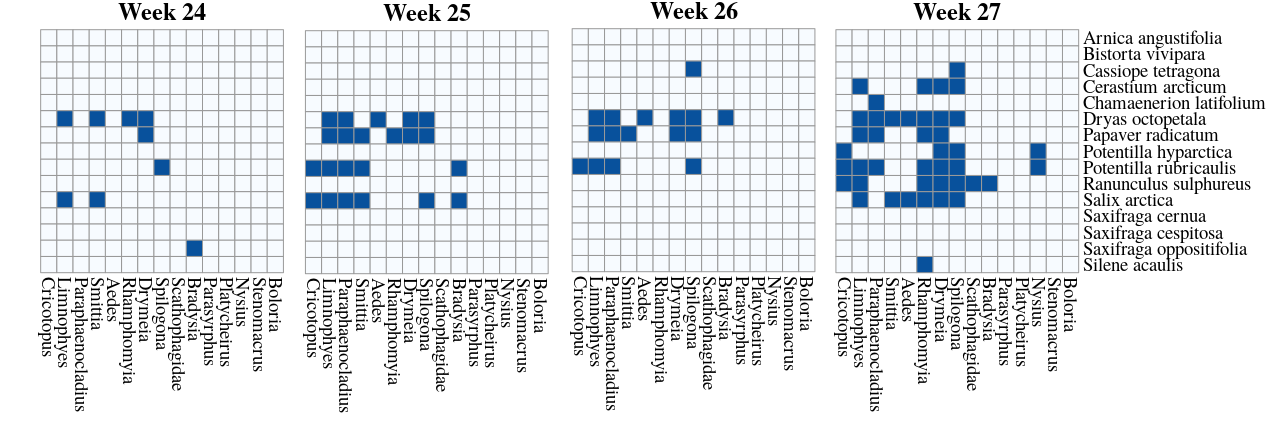}
    \caption{Evolution of pollination interactions among $15$ common plants and insects for the first four weeks of the $2010$ summer season. Rows indicate plant taxa, columns indicate insect taxa, and blue squares indicate observation of a pollination interaction.}
    \label{fig:1}
\end{figure}

Even in the static case, the problem of flexible modeling of network data is challenging, particularly when data are sparse and the number of nodes is moderate to large, as is the case in the present application. Indeed, sparse data are common in many application areas due to the difficulty of measuring real world networks. In the ecological context, painstaking field work is required to directly observe inter-species interactions: such interactions are often fleeting, requiring great effort to observe for any fixed pair, while there are many pairs of species with the potential to interact. While sparsity is a common practical concern, high dimensionality is an intrinsic concern, as the parameters of interest are pairwise edge probabilities resulting in an $N \times M$ and $N \times M \times T$-dimensional parameter space in the static case and dynamic case, respectively. Hence, large and sparse networks are standard in the field, and dimensionality reduction is a key component of any statistical analysis.
 
One common approach to dimensionality reduction assumes that edge probabilities depend only on group membership of the nodes. The simplest version of this approach assumes that grouping can be done based on observed node covariates, while allowing differential between- and within-group connectedness
(\cite{Wang.etal:1987}). Stochastic block models, instead, allow the node group assignments to be unknown
\citep{Nowicky.etal:2001,Kemp.etal:2004}, providing a class of models for inferring node communities. Such models assume stochastic equivalence: edges in the network are conditionally independent given the community memberships of the nodes. Although there is a rich literature on extensions, including the ability to allow node heterogeneity \citep{Airoldi.etal:2005, Karrer.etal:2010}, the overarching assumption of conditional independence given a single node latent class assignment makes stochastic block models incapable of accurately characterizing many real-world networks.

 A flexible alternative is motivated by homophily, the phenomenon in which nodes with similar characteristics are more likely to be linked.  \cite{Hoff.etal:2002} introduce the idea of representing the affinity of a node for interaction via a vector of latent features. The probability of an edge between the nodes $i$ and $j$ is a function of $||u_i-u_j||$, where $\{u_1,\hdots,u_n\}$ are low-dimensional vectors of unobserved latent characteristics. Although this class is well-suited to networks characterized by homophily, it is less able to capture other structures. To obtain a flexible generalization, 
 \cite{Hoff:2007} proposes a class of latent factor models in which the presence of an edge within a dyad is modeled as a function of $u_i^\top \Lambda u_j$ where $\Lambda$ is a diagonal matrix; this class of models is able to capture both stochastic equivalence and homophily. Moreover, given the presence of the matrix $\Lambda$, this model can also capture antihomophily ---- an opposite--attract phenomenon. 
 
The static latent space model of \cite{Hoff.etal:2002} has been extended to the dynamic case via a Markov process (\cite{Sarkar.etal:2005}) and random walk trajectories (\cite{Sewell.etal:2015}). In a parallel thread, \cite{Durante.etal:2014} use Gaussian processes for node-specific latent variables to extend \cite{Hoff:2007} to the time series setting, while \cite{Durante.etal:2015} use the same approach but with more flexible nested Gaussian processes allowing variability over time in the rate of change in the network structure. Stochastic block models have also been extended to the dynamic case: \cite{Xu.etal:2013} model discrete time dynamic networks through a state-space model, which combines a stochastic block model with a simple model for state evolution. \cite{Olivella.etal:2022} combine a hidden Markov model with a mixed membership stochastic block model to identify latent groups underlying the network structure.

Although the above developments have proven useful in their respective applications, dynamic latent factor models \citep{Sarkar.etal:2005, Durante.etal:2014,Sewell.etal:2015} are natural from an ecological perspective. Difficult-to-observe covariates drive interactions: features for insect species indicate latent consumption preferences, while features for plants indicate flower
traits, such that the inner product of a dyad's latent features measures alignment of insect preferences and plant traits, and hence their propensity for interaction. An example is a physical characteristic of the consumer (insect) affecting its access to resources (plants) of different shapes \citep{Garibaldi}. Because these features can be expected to change both seasonally with phenological patterns and in the long term as species adapt to changing climatic conditions, a smooth dynamic process on latent features can capture network evolution.

Despite the mechanical appeal, existing dynamic latent space approaches suffer from computational and inferential shortcomings in moderate to large networks with high sparsity and/or missingness. If we assume a $H$-dimensional latent space and model each component of the latent space through $(N+M)H$ independent temporal Gaussian processes as in \cite{Durante.etal:2014}, we have $(N+M)HT$ unknowns characterizing the dynamically changing edge probability matrix. This results in an excessively flexible model that is subject to overfitting, as illustrated in Section \ref{sec:real:app}.
 If in contrast, we utilize a dynamic stochastic block model as in \cite{Xu.etal:2013} and \cite{Olivella.etal:2022} when interaction propensities are truly driven by more complex latent features, the number of communities needed to effectively capture network transitivity explodes, resulting in an unfavorable trade-off between accuracy and interpretability.

Motivated by ecological applications, we present a novel Bayesian nonparametric tensor decomposition for bipartite dynamic networks. We build on the latent space approach of \cite{Hoff:2007} while achieving substantial dimension reduction through low-rank Candecomp/Parafac tensor decompositions (\cite{Carroll.etal:1970}; \cite{Harshman:1970}) of latent feature tensors. By characterizing each node attribute through low-dimensional exemplar traits, weights for each latent space dimension, and exemplar curves characterizing time evolution, we achieve drastic dimensional reduction while maintaining substantial flexibility. Although the focus is on bipartite networks with two types of nodes, the methodology can be trivially adapted to the case in which all nodes are of one type, as detailed in Section \ref{sec:model:formulation}.

\section{Dimension-reduced latent space model}
\label{sec:2}

\subsection{Model formulation} \label{sec:model:formulation}
Let $\A$ be the $N \times M \times T$  tensor with frontal slices $\A_{..t} \in \{A_1, \ldots, A_{T}\}$, encoding repeated observations of an undirected binary adjacency matrix. In this bipartite context, $i = 1, 2, \ldots, N$ and $j = 1, 2, \ldots, M$ represent nodes of two different types. The entries of $\A$ record interclass interactions such that $A_{ijt} = 1$ if a link is recorded between nodes $i$ and $j$ in observation $t$, and otherwise $A_{ijt} = 0$. Hence, the tensor $\A$ describes a multilayer network: the graph corresponding to $A_{..t}$ has an edge between vertices $i,j$ in layer $t$ if and only if $A_{ijt} = 1$. We suppose in the description to follow that $t$ indexes temporal observations, so that $\A$ encodes a dynamic network; however, the model easily accommodates spatial and spatio-temporal repetitions. For notational efficiency, we define $[N]$ as refering to the index set for the nodes of the first class $\{1,\ldots,N\}$, and define $[M]$ and $[T]$ similarly for nodes of the second class and the time points, respectively. 

We model edges in this dynamic network as the realization of conditionally independent Bernoulli random variables with $N\times M \times T$ latent probability array $\Pi$: 
$A_{ijt}| \Pi_{ijt} \sim \text{Bern}(\Pi_{ijt}).$
Dimensionality reduction and borrowing information between nodes and time points are needed to infer $\Pi$. We perform a first dimension reduction by projecting each node into a low-dimensional latent space. The propensity for interaction between two nodes at a given time is modeled as the weighted inner product of their $H$ dimensional latent traits ($X_{it}$ and $Y_{jt}$) plus a time-varying intercept ($\mu_t$), controlling overall connectivity.
This yields the following model: 
\begin{gather}
\label{lik}  \Pi_{ijt} = f^{-1}(S_{ijt}),\quad 
S_{ijt} = \mu_t + X_{it}^\top\Lambda^H Y_{jt}, \nonumber 
\end{gather}
where $f$ denotes the logit link function, and $\Lambda^H = \hbox{diag}(\lambda^H)$ provides weights for each dimension of the latent space and hence the influence of the corresponding features on the interaction propensity. In the ecological context, $\mu_t$ measures the impact of seasonality on pollination activity, while latent factors represent latent preferences and traits of insects and plants, respectively. These can be expected to progress over time in correspondence to seasonal shifts for each species, for example, changes in activity level and biological needs of insects as their life cycle progresses, and changes in scent emission and density of blooms from early emergence to peak, and decline.

We find that the above formulation is still too richly parameterized in many applications, motivating further dimension reduction. Let $\X \in \R^{NHT}$ be the node attribute tensor for nodes of the first class, with Mode-2 (row) fibers $\X_{i.t} = X_{it}$ providing the embedding of the latent trait space for node $i$ at time $t$; define $\Y \in \R^{MHT}$ similarly to be the node attribute tensor for the second class. The proposed Nested Exemplar Space (\texttt{NEX}) model achieves dimension reduction via low rank factorizations of these node attribute tensors, separately decomposing $\X$ and $\Y$ into the outer product of low dimensional exemplar traits for each node, exemplar curves characterizing temporal variation in the latent attributes, and  weights for each exemplar space dimension: 
\begin{gather}
\label{nex:dec} \X =\sum_{k=1}^K \lambda^K_k U^X_k\otimes V^X_k\otimes W^X_k \\ \label{eq:y} \Y =\sum_{k=1}^K \lambda^K_k U^Y_k\otimes V^Y_k\otimes W^Y_k,
\end{gather}
where $K$ is the dimension of the exemplar space, and $\lambda^K$ is a weight vector that plays a role similar to that of $\lambda^H$ in the latent trait space. For each dimension $k = 1, \ldots, K$ the component vectors $U_k^X \in \R^{N}$ and $U_k^Y \in \R^M$ embed the nodes of each type in the exemplar space, while $V_k^X \in \R^H$ link the exemplar and latent trait spaces, and $W_k^X, W_k^Y \in \R^T$ characterize evolution in the exemplar space over time. With this construction, we massively reduce the number of parameters to $\mathcal{O}(N+M+T)$, while maintaining sufficient flexibility to capture broad probability tensors $\Pi$, as we will show in the theoretical results presented in the Appendix \ref{SM_theory}. 

A key innovation of this model is the use of two separate latent dimensions, which allow highly flexible and interpretable modeling: the $H$-dimensional latent trait space and the $K$-dimensional exemplar space.  Let $\mathcal{S} \in \R^{NMT}$ be the latent interaction propensity tensor with entries $S_{ijt}$ indicating the latent log-odds of interaction between species $i$ and $j$ at time $t$, as introduced in \eqref{lik}. The dimension $H$ upper bounds the rank of every frontal slice of $\mathcal{S}$, and as a result controls the complexity of the networks at any time step $t$. Similarly, $K$ provides an upper bound on the rank of tensor decomposition; it follows that $H\leq K$. Details on the selection of appropriate values for $H,K$ are provided in Section \ref{subsec:prior} below.

Although the above model specification is developed for the bipartite case, this framework is readily extendable to cases with a single class of nodes. In this case, the dynamic network tensor $\A$ is semi-symmetric and we utilize a single latent factor tensor $\X$ to characterize nodes, defining the log-odds connection matrix for a given time step $t$ as 
\begin{equation*}
    \label{symm:ext}
    S_t = \mu_t \1_N \1_M^\top  + X_t \Lambda^H  X_t^\top.
\end{equation*}
Additionally, while we focus on applications with a finite number of time-points, our model naturally extends to continuous time. In the continuous-time case, the factor matrices $W^X$ and $W^Y$ characterizing the third mode are replaced with $K$ functions defined over time. 

With the theoretical results to follow, we characterize the class of time-varying latent probability matrices $\Pi_t$ with representation given by the inner product formulation of \eqref{lik} and the decomposition of the respective latent trait tensors given by \eqref{nex:dec}. Proposition 1 states that given $H,K$ sufficiently large, any time-varying latent interaction propensity matrix $S_t$ has such a representation. Detailed proofs are provided in Appendix \ref{SM_theory} in the continuous time setting and with $N = M$. Because these results immediately generalize to the discrete-time setting with generic $N \neq M$, we continue with the notation developed above. 

\begin{proposition}
\label{prop:1}
Given a time-evolving matrix $S_t\in\R^{N\times N}$, with $t\in[T]$, there exist finite integer values $H, K$ such that for every $t$ we have that
\begin{equation*}
S_t = \mu_t\1_N\1_M^\top + X_t Y_t^{\top},
\end{equation*}
where $\mu_t$ is a time-varying intercept and $X_t, Y_t$ are matrices of time-evolving latent vectors 
$$X_t = \sum_{k=1}^K  w^X_{kt} U^X_k \otimes V^X_k, \;\; Y_t = \sum_{k=1}^K w^Y_{kt} U^Y_k \otimes V^Y_k,$$
for $U^X_k$, $U^Y_k\in\R^N$ and $V^X_k$, $V^Y_k\in\R^H$.
\end{proposition}

Corollary 1 extends this result to the latent probability matrix $\Pi_t$. 

\begin{corollary}
\label{cor:1}
    Given a time-evolving link probability matrix $\Pi_t\in\R^{N\times N}$, with $t\in[T]$, there exist finite integer values $H, K$ such that for every $t$ we have that
    $$ \Pi_t = f(S_t), \;\;t\in[T],$$
    where $S_t$ is defined in Proposition \ref{prop:1}, and 
    $f(\cdot)$ is the logit function applied elementwise to $S_t$.
\end{corollary}

\subsection{Prior Specification}
\label{subsec:prior}
To complete our Bayesian model formulation, we place independent Gaussian priors on the latent features for each node and for each dimension of the exemplar space:

 \begin{gather}
 \label{eq:prior:1}   U^X_{ik} \sim N(0, \sigma^2),\:\;\text{ for all } ik\in[N]\times [K] \\
    U^Y_{jk} \sim N(0, \sigma^2),\:\;\text{ for all } jk\in[M]\times [K]\\
    V^X_{hk},V^Y_{hk} \sim N(0, \sigma^2),\:\;\text{ for all } hk\in[H]\times [K]  
 \end{gather}
For both the third mode factor matrices $W^X$ and $W^Y$ and the time-varying intercepts $\mu$, we utilize multivariate Gaussian priors, inducing temporal dependence in the networks: 
\begin{gather}
\label{first:prior}  W^X_k,W^Y_k \sim N(0, \Sigma)\;\;\text{ for all } k\in [K]\\
\label{last:prior}    \mu \sim N(\mu_0, \Sigma_{\mu}),
\end{gather}
where $\Sigma$ and $\Sigma_{\mu}$ characterize the temporal dependence.
For example, a reasonable choice is 
$\Sigma = \exp(-kD),$
where $D$ is a matrix quantifying distance between different observation times, and $k$ is a smoothing factor governing the rate of decay. For more complex study designs, with spatio-temporal structure or a temporal structure involving seasonality as well as annual change, a simple modification is to employ a separable covariance 
$\Sigma = \exp(-k_1 D_1) \exp(-k_2 D_2),$
where $D_1$ and $D_2$ indicate distance in separate dimensions. \eqref{first:prior}-\eqref{last:prior} can accommodate continuous time by substituting multivariate Gaussian distributions with Gaussian processes. 

Appropriate selection of dimension for the latent trait space and exemplar space is crucial to model performance, however, identifying the rank of a tensor is an NP-hard problem \citep{Haastad:1989}. We sidestep this issue by taking the over-fitted Bayesian modeling approach of \cite{Bhattacharya.etal:2011}, choosing $H,K$ as conservative upper bounds while adaptively shrinking redundant dimensions. Specifically, we place a double-shrinkage prior on  $\lambda^H, \lambda^K$, which collapses the distribution of factors in redundant dimensions to zero; this induces low rank in both the latent $H$- and exemplar $K$- spaces. The proposed shrinkage process is given by 
\begin{gather}
\label{mgp:prior}
    \lambda^K_k = \prod_{s=1}^k \theta^K_s,\;\; \theta^K_1\sim \hbox{Ga}(\alpha^K_1, \beta^K_1),\;\theta^K_{s>1}\sim\hbox{Ga}(\alpha^K_2,\beta_2^K) \\
    \label{mgp:prior:2}
    \lambda^H_h = \prod_{s=1}^h \theta^H_s,\;\; \theta^H_1\sim \hbox{Ga}(\alpha^H_1, \beta_1^H),\;\theta^H_{s>1}\sim\hbox{Ga}(\alpha^H_2,\beta_2^H).
\end{gather}
We can compute the first two moments of the exemplar space weights $\lambda_k^K$ as:
\begin{gather*}
	\nonumber E(\lambda^K_k) = \frac{\alpha^K_1}{\beta^K_1} \bigg( \frac{\alpha^K_2}{\beta^K_2}\bigg)^{k-1} \\
    	\hbox{var}(\lambda^K_k) = \bigg(\frac{\alpha^K_1}{\beta^K_1}\bigg)^2\bigg(\frac{\alpha^K_2}{\beta^K_2}\bigg)^{2(k-1)}\bigg\{ \bigg(1+\frac{1}{\alpha^K_1}\bigg)\bigg(1+\frac{1}{\alpha^K_2}\bigg)^{k-1} - 1\bigg\},
\end{gather*}
for $k = 1,\hdots,K$. We see that $\alpha^K_1$ and $\beta^K_1$ control the general level and variability of the entries in $\lambda^K$, while $\alpha^K_2$ and $\beta^K_2$ impose shrinkage.  Selecting $\alpha_1^K/\beta_1^K>1$ and $\alpha_2^K/\beta_2^K<1$  favors stochastically decreasing values for $\lambda_k^K$ as $k$ increases, which causes the contribution of higher indexed dimensions to decrease to zero.  An analogous selection is made for 
 hyperparameters in the prior of $\lambda_H$. Furthermore, since by construction $H\leq K$, we choose $\alpha_2^K / \beta_2^K > \alpha_2^H / \beta_2^H$ to favor a greater relative shrinkage of the latent trait space compared to the exemplar space. This soft constraint on relative dimensionalities yields good recovery of the true tensor rank in simulations (see Section \ref{SimulationNex}) while avoiding degeneracies induced by a hard constraint. 
 
 The above approach is related to the multiplicative gamma process 
\citep{Bhattacharya.etal:2011}; 
 see \cite{Durante.etal:2014} and \cite{DURANTE_MGPnote} for a detailed discussion of the selection of hyperparameters. Alternative shrinkage priors are possible, including the cumulative shrinkage prior of \cite{LegramantiCSP}. In practice, we find that the multiplicative gamma approach yields better recovery of true network rank in simulation studies.  

 Although the theoretical results in Proposition \ref{prop:1} and Corollary \ref{cor:1} guarantee that NEX is flexible enough to represent any dynamic network given sufficiently large $H,K$, 
  this result is only relevant from a Bayesian perspective if the priors selected have full support. We now turn our attention to this question and show that there is a positive probability of generating $\{\Pi_t, t \in [T] \}$  arbitrarily close to any true $\{\Pi^0_t, t\in[T]\}$. As with the previous theoretical results, detailed proofs are provided in Appendix \ref{SM_theory} in the more general continuous-time setting, and we adapt the statements to the discrete-time setting for notational consistency here.
Theorem \ref{thm:1} below states that \eqref{eq:prior:1}-\eqref{last:prior} has full support for latent interaction propensity matrices $S_t$.

\begin{theorem}
\label{thm:1}
    Let $p_{\mathcal{S}}$ be the prior induced on $\{S(t)\}$, $t\in [T]$ by the priors on the latent factors \eqref{eq:prior:1}-\eqref{first:prior} and time-varying intercept \eqref{last:prior}. 
    For every real matrix $S^0_t$ with $t \in [T]$ and for every $\epsilon>0$, 
    $$ p_{\mathcal{S}}\Big\{\underset{t\in [T]}{\max}\;\; 
     || S_t - S^0_t||_2 < \epsilon \Big\} > 0.$$
\end{theorem}

Corollary \ref{cor:2} extends this result to the time-varying latent probability matrix. 

\begin{corollary}
    \label{cor:2}
    Let $p_{\Pi}$ be the prior induced on $\{\Pi(t)\}$, $t\in [T]$ by the priors specified on the latent factors \eqref{eq:prior:1}-\eqref{first:prior} and time-varying intercepts \eqref{last:prior}. Then for every real probability matrix $\Pi^0_t$ with $t \in [T]$ and for every $\delta>0$  we have
    $$ p_{\Pi}\Big\{\underset{t\in [T]}{\max}\;\; 
     || \Pi_t - \Pi^0_t|| < \delta \Big\} > 0.$$
\end{corollary}

\subsection{Computation}
Because all of the parameters in \eqref{eq:prior:1}-\eqref{last:prior} are continuous, we obtain posterior samples using Hamiltonian Markov chain Monte Carlo implemented in the Stan probabilistic programming language. We find this works well in a wide variety of simulations and outperforms data augmentation Gibbs samplers that we considered in small scale settings. 

\section{Simulation study} 
\label{SimulationNex}
In this section we provide a simulation study to evaluate the performance of the proposed model in recovering the latent probabilities and data generating process. Performance is assessed via accuracy in inferring the effective dimensions of the latent and exemplar spaces, as well as in- and out-of-sample recovery of link probabilities. Further simulation studies demonstrating robustness to model miss-specification for data generated from dynamic latent factor and stochastic block models are provided in Appendix \ref{SM_SimulationNex}.

We generate a collection of $20 \times 25$ time-varying 
bipartite networks $A_t$ with $t\in\{1,2,\ldots,40\}$. Data are simulated according to our generative model with true exemplar and latent dimensions $K^{true}=5, H^{true}=2$ with $\lambda_H=\lambda_K = 1$ such that all dimensions contribute equally to the network structure a priori. We set the length scale of the Gaussian processes using $k_\mu=k_W=0.04$, and use a prior mean of $\mu_0=-0.6$ for the logit intercept and $\sigma^2=2/\sqrt{K^{true}}$ as prior variance for the entries of the non-temporal factor matrices in the exemplar space decomposition. The adjacency tensor $\mathcal{A}$ is generated by likelihood \eqref{lik}. To evaluate out-of-sample performance, we treat the final network $A_{40}$ as missing. We fit the proposed model to the synthetic data $\mathcal{A}$ using conservative truncation levels $H=K=10$; further details are given in Appendix \ref{SM_SimulationNex}. 

\begin{figure}[h!]
	\begin{minipage}{0.5\textwidth}
		\centering
		\includegraphics[width=\textwidth]{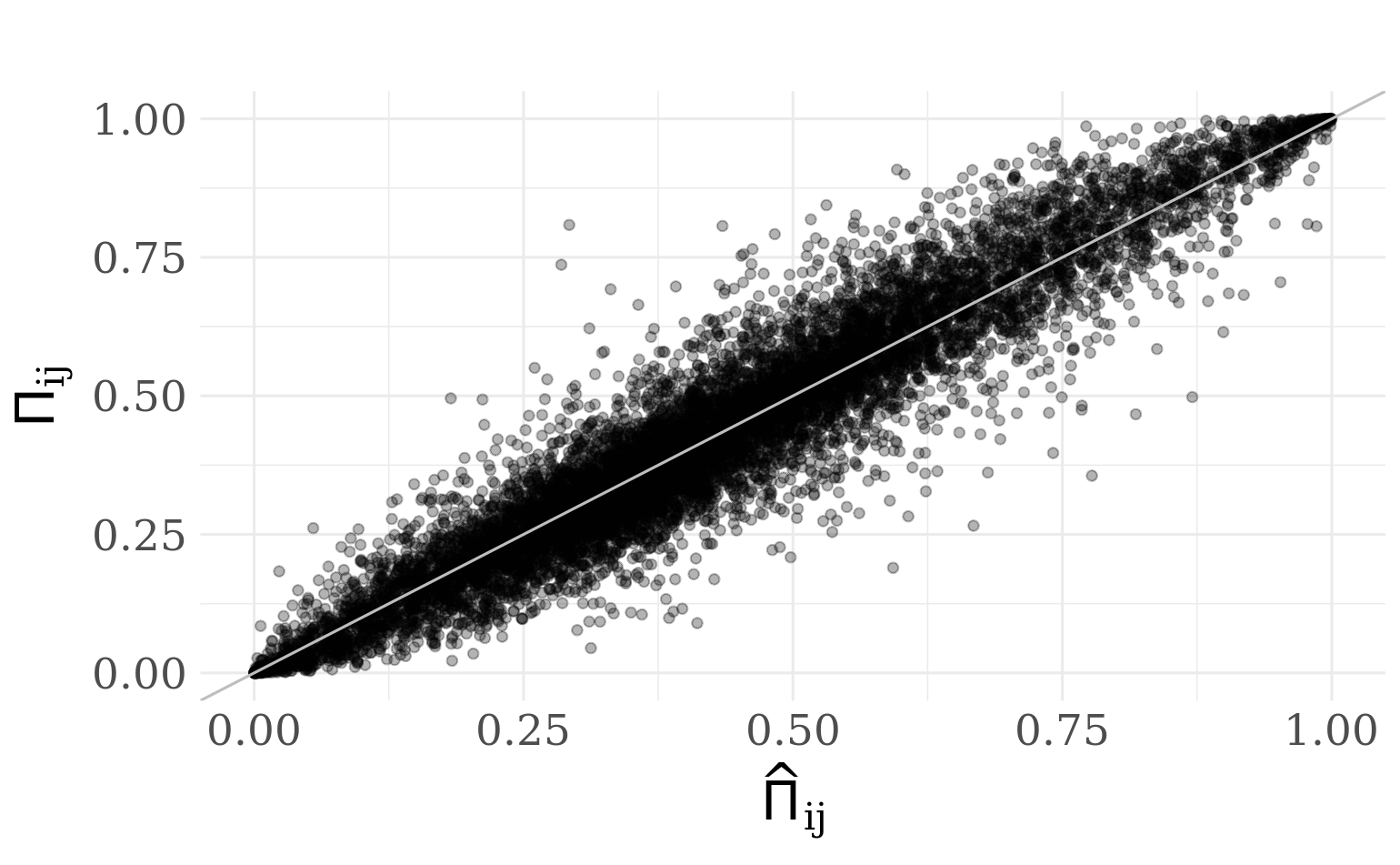}
		\label{fig:figure1}
	\end{minipage}\hfill
	\begin{minipage}{0.5\textwidth}
		\centering
		\includegraphics[width=\textwidth]{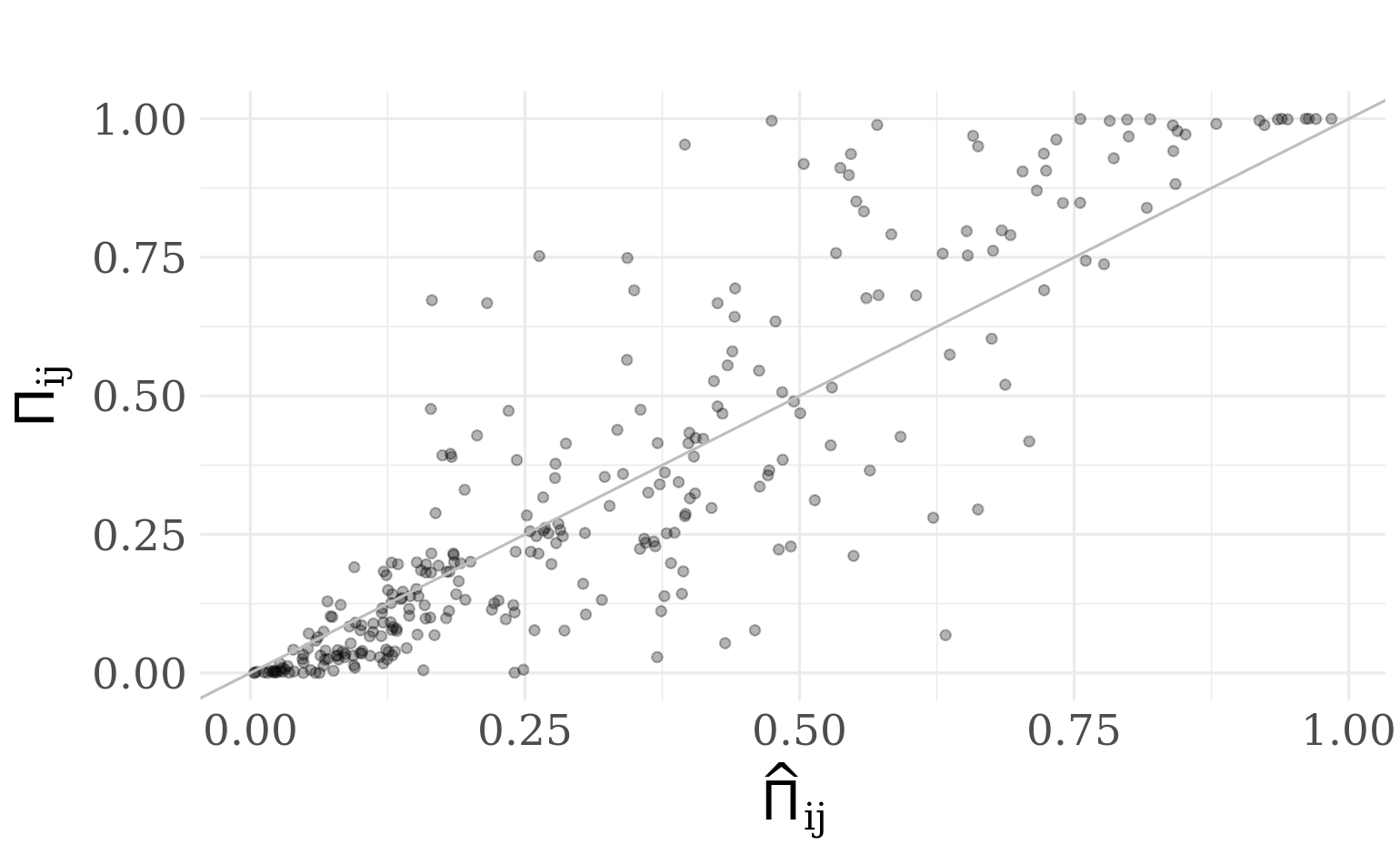}
		\label{fig:figure2}
	\end{minipage}
	\caption{Simulation results. Left: in-sample true probabilities $\pi_{ijt}^0$ vs posterior means $\hat\pi_{ijt}$ for $t\in\{1,2,\ldots,39\}$. Right: out-of-sample true probabilities $\pi_{ij40}^0$ vs posterior means $\hat\pi_{ij40}$. }
	\label{fig:prob:is:vs:oos}
\end{figure}

Our model shows excellent performance in estimating the link probabilities both in- and out-of-sample. 
The estimated posterior mean probabilities capture the variability observed in the true values of $\Pi_{t}$ for $t\in\{1,\ldots,40\}$, with only a slight drop in performance for the out-of-sample network at $t=40$. Figure \ref{fig:prob:is:vs:oos} presents an aggregate analysis across time points comparing estimated and true link probabilities in-sample and out-of-sample. Figure \ref{fig:prob:ts} shows estimated and true link probability matrices for $t \in \{1,20,40\}$.

\begin{figure}[htp]
    \centering
    \includegraphics[scale=0.25]{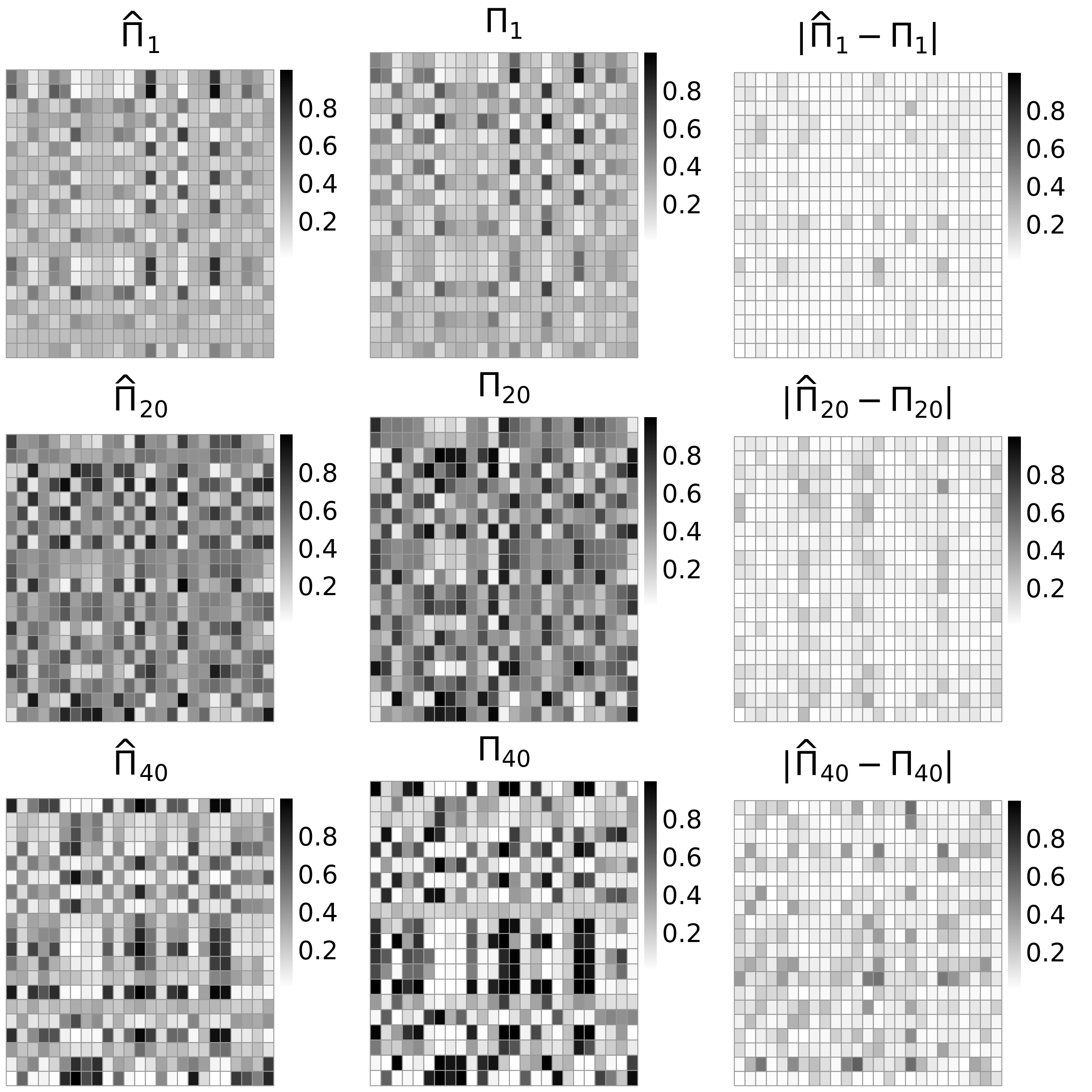}
    \caption{Simulation results. From top to bottom: rows indicate the networks at randomly chosen in-sample time steps $t = 1, 20$ and out-of-sample time step $t=40$. \emph{Left column}: Posterior means of $\hat\Pi_{t}$. \emph{Center column}: True underlying $\Pi_{t}$. \emph{Right column}: Absolute value of the difference between the two $|\hat\Pi_{t}-\Pi_{t}|$.}
    \label{fig:prob:ts}
\end{figure}

 Figure \ref{fig:shrink} shows that the model is capable of effectively shrinking towards the true underlying values of $H,K$ by collapsing the distribution of factors in redundant dimensions to zero. This greatly simplifies model specification, allowing the user to set a conservative upper bound for these dimensions. 
 
\begin{figure}[htp]
    \centering
    \hspace{-1.0cm}
	\includegraphics[scale=0.35]{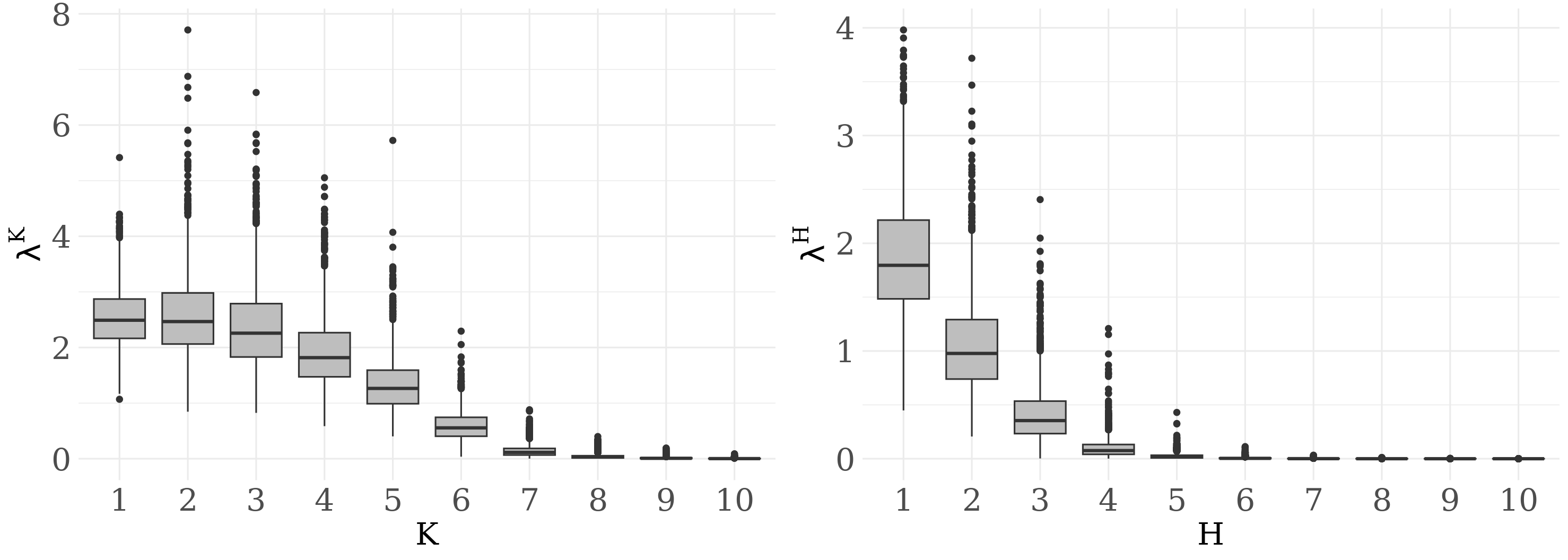}
	\caption{Posterior samples for different components of $\lambda^K$(left) and $\lambda^H$(right) when data are simulated with true values $K = 5$ and $H = 2$}
	\label{fig:shrink}
\end{figure}

\section{Application to dynamic plant-pollinator networks in the arctic}
\label{sec:real:app}

Interaction networks offer a concise overview of complex ecological communities while providing insights into a community's resilience in the face of exogenous shocks. We analyze data consisting of repeated observations of a plant-pollinator network in the high-arctic Zackenberg Valley in Northeast Greenland National Park. Observations were collected on a 500 by 500 meter study plot weekly in the snow-free season in the years 1996, 1997, 2010, 2011, and 2016, resulting in a collection of 49 networks, published in \cite{Cirtwill}. Because network observations were collected over a two-decade period in which climatic changes have been particularly pronounced in the Arctic \citep{Rantenen2022}, it is of great interest to analyze the effects of shifting phenologies (that is, the timing of flower blooms and the emergence of arthropods) on pollination interactions. The observed networks are highly sparse, with less than one percent of all possible interactions observed; additionally, 46\% of observed interactions are observed only once.

The observed networks encode flower visitation interactions over a 40-minute period: two individuals of each species were identified and monitored, with all insects visiting flowers of the selected individuals recorded, captured and identified (with taxa confirmed by DNA barcoding \citep{Hebert2003}). Although flower visitation data were collected for all study years, in 2016 the data were complemented with pollen transport data, based on microscopic analysis of pollen grains present on each insect captured during flower visitation (see \cite{Cirtwill2024}). Hence, the pollen transport data yield a second network, which can be used to validate predictions based on the flower visitation data. In total, the data include 39 plant taxa and 114 arthropod taxa. Since most arctic pollinators are small, dark-colored, and difficult to distinguish even by experts, taxa are recorded with some taxonomical ambiguity, with some individuals identified only at the genus or family level and others fully identified to the species level. Additional details on the Zackenberg data can be found in Appendix \ref{SM_app} and \cite{Cirtwill}.

We resolve taxonomical ambiguity by collapsing partially identified individuals into operational taxonomical units, and limit our analysis to taxa occurring in more than one interaction across all time periods. Additionally, we remove two weekly networks consisting of interactions involving a single plant present. Further details on taxonomical cleaning and subsetting can be found in Appendix \ref{SM_app}.  The application of the protocol described above results in a network of 35 plant taxa and 58 insect taxa observed over 47 time periods with an observed interaction prevalence of 1.8\%. However, behind this ostensible sparsity, there is a high degree of missingness due to nonobservable interactions: for an interaction to be observed, the involved species must co-occur in the observation period. 

To avoid having variation in species occurrence drive our inferences,  we model network variation conditional on species co-occurrence.
Plant occurrence is measured without error, while insect occurrence is inferred from the interaction data: we define insects that occur at time $t$ as those observed in at least one interaction at time $t$ or the preceding week $t-1$ or the following week $t+1$. This is because insect phenology is strongly structured by season (see \cite{Wirta}, \cite{Ascanio2024}; see also Figure 11 in the Appendix), with the overall activity period defined by a start date and an end date. During the intervening period, the taxon will be present and forage for at least some portion of the midday observation window, whereas some taxa may be missed due to under-sampling. We performed a sensitivity analysis with respect to the construction of the occurrence matrices in Appendix \ref{SM_app_sens}. 

We adopt a conditional modification of \texttt{NEX} for the Zackenberg data. Let $A_{ijwr}$ be an indicator for the event that we observe that a plant $i$ is pollinated by an insect $j$ during the seasonally aligned week $w$ of year $r$. We define a model for interaction conditional on pairwise co-occurrence:  
\begin{align*}
     E(a_{ijwr} \mid \pi_{ijwr},  O_{ijwr} ) &= \pi_{ijwr} O_{ijwr} = \frac{ O_{ijwr}}{1 + e^{-s_{ijwr}}} \\
    s_{ijwr} &= \mu_w + \alpha_i + \gamma_r + \beta c_{wr} +
     x_{iwr}'\Lambda^H y_{jwr},
\end{align*}
where $\mu_w, \gamma_r, \alpha_i$ are random effects for seasonality, annual variability, and plant node degree heterogeneity, $\beta$ is a fixed effect for temperature in degrees Celsius ($c_{wr}$), and $O_{ijwr}$ is an co-occurrence indictor with $O_{ijwr}=1$ if plant $i$ and insect $j$ co-occur in week $w$ of year $r$. 
Motivated by this two-dimensional time structure in which networks can be located temporally by both the year and seasonality week, we structure autocorrelation using a separable covariance with seasonality and year dependence. The random and fixed effects are given weakly informative normal priors; details on the full model specification are provided in Appendix \ref{SM_app_priors}. 

The model above is fit using the proposed framework and a competitor provided by the dynamic latent factor model of \cite{Durante.etal:2014} using the latent and exemplar space dimension $K = H = 10$ for \texttt{NEX} and latent space dimension $H = 10$ for the dynamic latent factor model. 
To fit \texttt{NEX}, we generate $2000$ post-burn in posterior samples using Stan, while the same number of post-burn in samples are generated by fitting the dynamic latent factor model with the Polya-Gamma data augmentation Gibbs sampling scheme described in \cite{Durante.etal:2014}, modified to accommodate the bipartite data structure and relevant covariates. 

To compare the models, we perform 10-fold cross validation of the flower visitation data, and additionally utilize the pollen transport data for further out of sample validation. The existence of the pollen transport data from 2016 allows us to validate predictions made by the model trained on flower visitation data only adopting an approach similar to that utilized in \cite{Papadogeorgou2021}. We consider the posterior probability for all missing cells in the flower visitation network with corresponding edges in the pollen transport data (heldout edges), and the posterior probability for all such tuples that are missing in flower visitation network, but without observed interactions in the pollen transport network (heldout non-edges). The ratio of the mean posterior probability among heldout pollen transport edges to heldout pollen transport non-edges should be greater than one if the model successfully discriminates non-interactions from interactions.

We find strong evidence of overfitting for the dynamic latent factor model with superior performance of \texttt{NEX} in out-of-sample metrics. The dynamic latent factor model has a higher in-sample fit with a mean AUC of 0.98 versus 0.90 for \texttt{NEX} across cross validation replicates. However,
\texttt{NEX} has better out-of-sample performance as can be seen in Figure \ref{app_oos_boxplots}. While both models demonstrate the capacity to identify true interactions from the held-out pollen transport data, we see superior performance of \texttt{NEX} in recovering these held-out interactions, demonstrating that the more parsimonious model not only fits the original data better, but further generalizes to data collected using entirely different methods. This suggests superior ability of \texttt{NEX} to deliver generalizable insights from sparse data where other models over-fit to noise.

\begin{figure} 
	\begin{minipage}[t]{0.45\textwidth}
		\includegraphics[width=\textwidth]{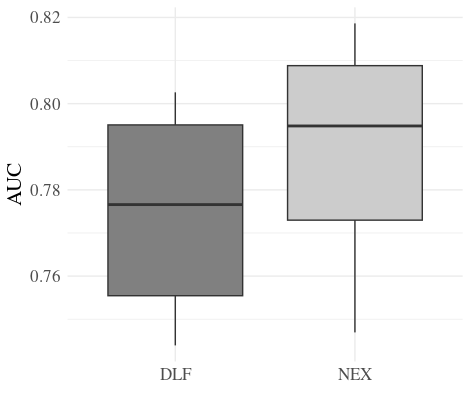}%
	\end{minipage}%
	\hspace{0.0\textwidth}
	\begin{minipage}[t]{0.45\textwidth}
		\includegraphics[width=\textwidth]{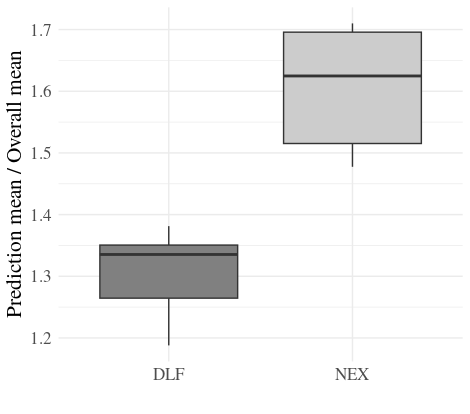}%
	\end{minipage}
	\caption{Cross Validation Results for the proposed Nested Exemplar Latent Space Model (\texttt{NEX}) and the dynamic latent factor model (\texttt{DLF}). We perform ten-fold cross validation and for each data set calculate the out-of-sample (AUC) and the ratio of the mean posterior interaction probability for held-out recorded pollen transport interactions to the mean posterior interaction probability for missing edges overall (right).}
    \label{app_oos_boxplots}
\end{figure}

\subsection{Model interpretation}

The application of \texttt{NEX} to the Zackenberg data yields a number of interesting insights into the stability of arctic pollination networks and the forces driving pollination activity. A key ecological question relates to the relationship between temperature and interaction intensity, with a positive relationship presumed based on a general trend of increasing interaction richness with temperature. However, in fitting the conditional \texttt{NEX} model, we find a small negative effect of temperature. This is likely the artifact of an influential outlier network from 2011 with a very high prevalence of interaction on an unseasonably cool day late in the season. Removing this term, the effect of temperature disappears, providing evidence that the impact of temperature on interactions works largely through species occurrence and not directly through pollination activity. Similarly, we find that a hypothesized seasonality effect producing greater interaction richness in the peak of summer disappears or even reverses when controlling for occurrence. 

To address potential concern that these results are driven by bias introduced by our occurrence construction, we perform sensitivity analysis under a variety of occurrence scenarios and find that the results persist; the details are provided in the Appendix \ref{SM_app_sens}. These findings are consistent with recent work showing the remarkable stability of arctic pollination networks over time \citep{Cirtwill}, and suggest that abiotic factors impact pollination activity primarily through occurrence. 

A crucial question in the study of pollination networks relates to the existence of unobserved interactions, both those that are occurring but not seen and those which are not yet occurring due to non-overlapping phenologies. Fitting \texttt{NEX} allows us to analyze posterior interaction probabilities across observed non-edges of both types. If we focus on potentially understudied insect families (defined to be those with limited pollination activity), we identify a large number of not previously observed high posterior probability interactions for the families Agromyzidae and Scathophagidae, both consisting of flies with a high potential to be missed by observing researchers, and Braconidae and Ichneumonidae, both parasitoid wasps where time spent on flowers limits time spent searching for hosts, making them harder to observe on the flowers. 

Finally, the model suggests a number of new pollination interactions which have not previously been possible due to non-overlapping phenologies. In particular, we find a high mean posterior probability of pollination of the arctic willow \textit{Salix arctica} by the very rare genus \textit{Atractodes}, given appropriate overlap of species occurrence. Additionally, we find a number of rare species to be likely but unobserved pollinators of the very briefly blooming \textit{Saxifraga hyperborea}. These results suggest potential redundancies in the pollination networks which can contribute to resilience of the arctic ecological community under destabilization from climate-induced phenology shifts. Importantly, they also provide testable hypotheses regarding unobserved but likely interactions, which can next be resolved by targeted experiments. See Appendix \ref{SM_app} for a detailed presentation of new high posterior probability interactions. 

\section{Discussion}

We introduce a novel class of models for dynamic networks, which dramatically reduce dimensionality relative to existing latent-space approaches. Although existing latent factor models borrow information across networks over time, they utilize independent Gaussian processes for each component of the latent space and each node in the network. In contrast, by utilizing a tensor factorization approach, the proposed model replaces these independent Gaussian processes with a small set of exemplar curves which are common across nodes and latent space components, yielding dramatic computational advantages and superior out-of-sample performance on both real-world and simulated networks. 

The application presented here demonstrates the broad suitability of the model to a number of challenging real-world network settings, which are often characterized by sparsity and observational bias. Careful modeling of a complex dynamic pollination network within the \texttt{NEX} framework delivers novel insights into the drivers of pollination activity, while suggesting future directions for research, including extension to networks with weighted or discrete-valued edges and to spatially-indexed replicates. Simultaneously, the computational efficiency of the model could facilitate a more in-depth treatment of observational bias, including formalizing uncertainty of co-occurrence through the addition a phenology sub-model, or the incorporation into more expansive joint species distribution models.  

\section{Acknowledgment}
This project has received funding from the European Research Council under the European Union’s
Horizon 2020 research and innovation programme (grant agreement No 856506) 
and from the Office of Naval Research (N00014-21-1-2510).
The authors express their gratitude to Otso Ovaskainen, Alyssa Cirtwell, and Bora Jin for their insightful suggestions and assistance with the complex Zackenberg data. 

\bibliographystyle{apalike}
\bibliography{paper-ref}

\begin{appendix}
\label{appendix}
\setcounter{theorem}{0}
\setcounter{proposition}{0}
\setcounter{corollary}{0}
\renewcommand{\theequation}{\Alph{section}.\arabic{equation}}

\section{Theoretical results} 
\label{SM_theory}
This section provides the proofs of the theorems, propositions and corollaries stated in Subsections \ref{sec:model:formulation}-\ref{subsec:prior} of the manuscript. These theorems prove two important results: first, we show that \texttt{NEX} is flexible enough to recover any continuously evolving dynamic network, and second, that the prior \eqref{eq:prior:1}-\eqref{last:prior} has full-support in that it assigns positive probability to arbitrarily small neighborhoods of any dynamic network process.

The following proofs are provided in the continuous time setting with $t\in\mathcal{T} \subset \R$ compact, but are easily adapted to the discrete time setting presented in the theoretical statements in Sections \ref{sec:model:formulation}-\ref{subsec:prior}. To highlight the continuous-time structure, we introduce a slight modification to the notation: a matrix at a given time $t$ is indicated with $S(t)$ as its entries are functions of $t\in\mathcal{T}$. For other components which depend on time, such as $W^X, W^Y$, we refer to component values with $w_k(t)$.  In the proofs that follow, we set the shrinkage parameters $\lambda_H = \lambda_K = 1$ and set $H$ and $K$ sufficiently large to allow recovery of any given tensor. Lastly, we assume, without loss of generality, that $N=M$ is such that there are an equal number of nodes of each type to simplify the notation. 

\begin{proposition}
\label{prop:1:cont}
Given a time-evolving matrix $S(t)\in\R^{N\times N}$, with $t\in\mathcal{T}$ compact, there exist finite integer values $H, K$ such that for every $t\in\mathcal{T}$ we have that
\begin{equation}
\label{eq:1:prop:1}
S(t) = \mu(t)\1_N\1_N^\top + X(t) Y(t)^{\top},
\end{equation}
where $\mu(t)$ is a time-varying intercept and $X(t), Y(t)$ are matrices of time-evolving latent vectors:
\begin{gather}
\label{eq:proof:1}    X(t) = \sum_{k=1}^K  w^X_k(t) U^X_k \otimes V^X_k  \\
\label{eq:proof:2}   Y(t) = \sum_{k=1}^K w^Y_k(t) U^Y_k \otimes V^Y_k,
\end{gather}
for $U^X_k$, $U^Y_k\in\R^N$ and $V^X_k$, $V^Y_k\in\R^H$, while $w^X_k(\cdot), w^Y_k(\cdot)$ are real scalar functions on $\mathcal{T}$.
\end{proposition}
\begin{proof}
Assume without loss of generality that $\mu(t)=0$ for every $t\in\mathcal{T}$. Consider the time evolving singular value decomposition of the matrix $S(t)$ with rank $H(t)$
\begin{equation*}
S(t) = L(t) D(t) R(t)^{\top},\qquad t\in\mathcal{T}
\end{equation*}
with $L(t), R(t)\in\R^{N\times H(t)}$, $D(t)\in\R^{H(t)\times H(t)}$. Define $H_0=\max_{t\in\mathcal{T}} H(t)$, which exists because $\mathcal{T}$ is compact and the network has finite dimension. Set $H\geq H_0$ and consider the block matrices
\begin{equation*}
X(t) = \{L(t) D(t)^{1/2}: 0_{N\times \{H-H(t)\}}\},\;\; Y(t) =  \{R(t) D(t)^{1/2}: 0_{N\times \{H-H(t)\}}\}.
\end{equation*}
It is straightforward to see that there exists $K$ such that \eqref{eq:proof:1}-\eqref{eq:proof:2} hold: let $e_s^D$ denote $s$-th canonical base in $\R^D$, and consider the set of rank-1 matrices $B^{nk} = e^N_n\otimes e^H_h\in\R^{N\times H}$. We can express $X(t)$ as
\begin{equation}
\label{eq:expanded}
X(t) = \sum_{n=1}^N\sum_{h=1}^H x_{nh}(t) B^{nh} = \sum_{n=1}^N\sum_{h=1}^H  x_{nh}(t) e^N_n\otimes e_h^H,
\end{equation}
which is equivalent to \eqref{eq:proof:1} with $k=N(h-1)+n$ and $K=NH$. We proceed in the same manner for $Y(t)$.
\end{proof}
\begin{remark}
\label{rem:1}
    If we have a discrete number $T$ of time points, \eqref{eq:proof:1}-\eqref{eq:proof:2} are equivalent to the two CP decompositions  \eqref{nex:dec}-\eqref{eq:y} for the tensors $\X,\Y$ presented in Section \ref{sec:2}. \eqref{eq:expanded} then represents the standard worst-case scenario upper bound on the rank $K$ for the CP decomposition on a generic tensor $\mathcal{X}\in\R^{N\times H\times T}$.
\end{remark}
This result for $S(t)$ directly extends to $\Pi(t)$, as we summarize below. 
\begin{corollary}
    Given a time-evolving link probability matrix $\Pi(t)\in\R^{N\times N}$, with $t\in\mathcal{T}$ compact, there exist finite integer values $H, K$ such that for every $t\in\mathcal{T}$ we have that
    $$ \Pi(t) = f\{S(t)\}, \;\;t\in\mathcal{T},$$
    where $S(t)$ is defined as in \eqref{eq:1:prop:1}-\eqref{eq:proof:2} and $f(\cdot)$ is the logit function applied elementwise.
\end{corollary}
\begin{proof}
    The proof follows immediately from Proposition \ref{prop:1} and from the fact that the logit link function $f(\cdot)$ is a continuous, one-to-one, and increasing function.
\end{proof}
Proposition \ref{prop:1} ensures that the continuous-time \texttt{NEX} decomposition \eqref{eq:proof:1}-\eqref{eq:proof:2} is flexible enough to represent any dynamic network given sufficiently large $H,K$. From a Bayesian perspective, this result is relevant only if the prior on the model parameters has full support, that is, that there is a positive probability of generating $\{\Pi(t), t\in\mathcal{T}\}$ that is arbitrarily close to any true $\{\Pi_0(t), t\in\mathcal{T}\}$. We set $H\geq N$ and $K\geq NH$ which is sufficient to allow recovery of any dynamically evolving matrix $\mathcal{S}_0$, as we see from the proof of Proposition \ref{prop:1}; these values are chosen to simplify the following proof, but as hinted in Remark \ref{rem:1} are in practice unnecessarily large.  The following Theorem \ref{thm:1} and Corollary \ref{cor:2} prove precisely the full support of the previous \eqref{eq:prior:1}-\eqref{last:prior}. As before, we first prove the complete support for the log-odds process $S(t)$, and then extend the result to $\Pi(t)$ through the continuity of the logit link $f(\cdot)$.
\begin{theorem}
\label{thm:1:cont}
    Let $p_{\mathcal{S}}$ be the prior induced on $\{S(t), t\in \mathcal{T}\}$ by the priors specified on the latent factors \eqref{eq:prior:1}-\eqref{first:prior} and the time-varying intercept \eqref{last:prior}. 
    If $\mathcal{T}$ is compact, then for every real matrix $S_0(t)$ with $t \in \mathcal{T}$ continuous, and for every $\epsilon>0$  we have
    $$ p_{\mathcal{S}}\Big\{\underset{t\in \mathcal{T}}{\sup}\;\; 
     || S(t) - S_0(t)||_2 < \epsilon \Big\} > 0.$$
\end{theorem}
\begin{proof} 
Since $\mathcal{T}$ is compact, for every $\epsilon_{0}>0$ there exists an open covering of $\epsilon_0$-balls \\ 
$B_{\epsilon_0}(t_0)=\{t: ||t-t_0 ||_2<\epsilon_0 \}$ with a finite subcover such that $\mathcal{T} \subset \cup_{t_0 \in \mathcal{T}_0}B_{\epsilon_0}(t_0)$, where $|\mathcal{T}_0|<\infty$. Then
$$ \Pi_{S} \Big( \sup_{t \in \mathcal{T}}||S(t)-S_0(t) ||_2<\epsilon \Big)= \Pi_{S} \Big(\max_{t_0 \in \mathcal{T}_0} \sup_{t \in B_{\epsilon_0}(t_0)}||S(t)-S_0(t) ||_2<\epsilon \Big).$$
Define $Z(t_0)=\sup_{t \in B_{\epsilon_0}(t_0)}||S(t)-S_0(t) ||_2$. Since
$$\Pi_{S} \big(\max_{t_0 \in \mathcal{T}_0} Z(t_0)<\epsilon \big)>0 \Longleftrightarrow \Pi_{S} \left(Z(t_0)<\epsilon \right) >0, \ \text{ for all } t_0 \in \mathcal{T}_0,$$
we only need to look at each $\epsilon_0$-ball independently. Using the triangle inequality twice we obtain a lower bound on the desired probability with 
\begin{equation}
    \Pi_{S} \left(||S(t_0)-S_0(t_0) ||_2\;+\sup_{t \in B_{\epsilon_0}(t_0)}||S_0(t_0)-S_0(t) ||_2\;+ \sup_{t \in B_{\epsilon_0}(t_0)}||S(t_0)-S(t) ||_2<\epsilon\right), \nonumber
\end{equation}
which then can be further decomposed as 
\begin{multline}
\label{fact}
    \Pi_{S}\left(\sup_{t \in B_{\epsilon_0}(t_0)}||S_0(t_0)-S_0(t) ||_2<\frac{\epsilon}{3}\right) \\
\Pi_{S}\left( \sup_{t \in B_{\epsilon_0}(t_0)}||S(t_0)-S(t) ||_2<\frac{\epsilon}{3} \;\Big\vert \;|| S(t_0) - S_0(t_0) ||_2 < \frac{\epsilon}{3}\right)
\Pi_{S} \left(||S(t_0)-S_0(t_0) ||_2<\frac{\epsilon}{3}\right)  
\end{multline}
since the first term states the continuity property for a deterministic component, which is independent from the other two terms, while the latter two follow from the definition of conditional probabilities. We now proceed by analyzing each term in \eqref{fact} separately. Based on the continuity of $S_{0}(\cdot)$, for all $\epsilon/3>0$, there exists an $\epsilon_{0,1}>0$ such that:
$$||S_0(t_0)-S_0(t_0) ||_2<\frac{\epsilon}{3}, \quad ||t-t_0 ||_2<\epsilon_{0,1}\implies \Pi_{S}\left(\sup_{t \in B_{\epsilon_{0,1}}(t_0)}||S_0(t_0)-S_0(t) ||_2<\frac{\epsilon}{3}\right)=1. $$
The second term in \eqref{fact} states the continuity property of $S(t)$ in a neighborhood of $t_0$, with the conditioning
 event restricting the analysis to the subset of all the realizations of $S(t)$ having $S(t_0)$ lying in a $\epsilon/3$-neighborhood of $S_0(t_0)$. We first prove the continuity property of $S(t)$ in its unrestricted sample space. The
 continuity in a subset will follow as a consequence. 
 
 Given the Gaussian process prior on the time dependent components of the \texttt{NEX} factors, 
 $$ [X(t) Y(t)^\top]_{ij} = \sum_{k,k'}^K w^X_{k}(t)w^Y_{k'}(t) \langle V_{k}^X, V_{k'}^Y\rangle u^X_{ik}u^Y_{jk'}$$
 represents a finite sum in which the temporal components are almost surely continuous and interact through a product; thus the resulting elements of the matrix are also almost surely continuous across time. Therefore, $S(t)=\mu(t) +X(t)Y(t)^\top$ is almost surely continuous on $\mathcal{T}$ since also the baseline $\mu(\cdot)$, given the Gaussian process prior, is itself almost surely continuous.
Hence again for all $\epsilon/3>0$, there exists an $\epsilon_{0,2}>0$ such that
$$\Pi_{S}\left(\sup_{t \in B_{\epsilon_{0,2}}(t_0)}||S(t_0)-S(t) ||_2<\frac{\epsilon}{3}\right)=1.$$
Since we proved that all realizations from $S(t)$ are continuous in a neighborhood of $t_0$, the same will
be true for the subset of the sample space induced by the condition $||S(t_0) - S_0(t_0)||_2< \epsilon/3$. 
To examine the last term, first note that
\begin{multline*}
\Pi_{S}\left(||S(t_0)-S_0(t_0) ||_2 <\frac{\epsilon}{3} \right)= \\
=\Pi_{S} \left(||\mu(t_0)\1_N\1_N^\top  +X(t_0) Y(t_0)^\top -\mu_0(t_0)\1_N\1_N^\top -X_0(t_0) Y_0(t_0)^{\top}||_2<\frac{\epsilon}{3}\right).
\end{multline*}
Thus, using triangle inequality, we can bound this probability by
\begin{eqnarray*}
&&\Pi_{S}\left(||S(t_0)-S_0(t_0) ||_2 <\frac{\epsilon}{3} \right)\\
&\geq& \Pi_{S} \left(||X(t_0)Y(t_0)^\top-X_0(t_0)Y_0(t_0)^{\top}||_2<\frac{\epsilon}{6}\right)\Pi_{\mu}\left(||\mu(t_0)\1_N\1_N^\top -\mu_0(t_0)\1_N\1_N^\top ||_2<\frac{\epsilon}{6}\right).
\end{eqnarray*}
Based on the support of the Gaussian prior,
$$\Pi_{\mu}\left(||\mu(t_0)\1_N\1_N^\top -\mu_0(t_0)\1_N\1_N^\top ||_2<\frac{\epsilon}{6}\right)=\Pi_{\mu}\left(|( \mu(t_0)-\mu_0(t_0)|<\frac{\epsilon}{6\sqrt{N}}\right)>0.$$
Given Proposition \ref{prop:1}, there exist $\epsilon_X, \epsilon_Y$ such that given
$$(||X(t_0)-X_0(t_0)||_2 < \epsilon_X) \cap (||Y(t_0)-Y_0(t_0)||_2 < \epsilon_Y),$$
we have $||X(t_0)Y(t_0)^\top-X_0(t_0)Y_0(t_0)^{\top}||_2 < \epsilon/6$. Given the prior independence of $X(t_0), Y(t_0)$ we can further lower bound the remaining term with
$$ p_X(||X(t_0)-X_0(t_0)||_2 < \epsilon_X)  p_Y(||Y(t_0)-Y_0(t_0)||_2 < \epsilon_Y). $$
Focusing on the first term, since $$X(t_0)=\sum_{k=1}^K w^X_{k}(t_0)U_k^X (V_{k}^X)^{\top}$$ is a sum of $K$ independent terms, we once again invoke the triangle inequality to obtain
$$ p_X(||X(t_0)-X_0(t_0)||_2 < \epsilon_X) \geq \prod_{k=1}^K p_X\left(|| w^X_{k}(t_0)U^X_k (V^X_k)^\top - w^X_{0,k}(t_0) W^X_{0,k}(V^X_{0,k})^\top||_2 < \frac{\epsilon_X}{K}\right).$$
Since $w^X_{k}(t_0) U^X_k (V^X_k)^\top$ has full support on the space of rank-1 $[N,H]$ matrices, we know that for every $k$
$$ p_X\left(|| w^X_{k}(t_0)U^X_k (V^X_k)^\top - w^X_{0,k}(t_0) W^X_{0,k}(V^X_{0,k})^\top||_2 < \frac{\epsilon_X}{K}\right) > 0.$$
Thus $p_X(||X(t_0)-X_0(t_0)||_2 < \epsilon_X) > 0$. Similarly, it can be shown that $p_Y(||Y(t_0)-Y_0(t_0)||_2 < \epsilon_Y)>0$. Additionally, given the previously demonstrated large support for the common mean $\mu(\cdot)$, we obtain 
$$\Pi_{S}\left(||S(t_0)-S_{0}(t_0) ||_2 <\frac{\epsilon}{3} \right) > 0.$$
To conclude, set $\epsilon_0= \min(\epsilon_{0,1}, \epsilon_{0,2})$, such that all terms in \eqref{fact} are positive. 
\end{proof}

\begin{corollary}
    \label{cor:2:cont}
    Let $p_{\Pi}$ be the prior induced on $\{\Pi(t), t\in \mathcal{T}\}$ by the priors specified on the latent factors \eqref{eq:prior:1}-\eqref{first:prior} and time varying intercepts \eqref{last:prior}. If $\mathcal{T}$ is compact, then for every real probability matrix $\Pi_0(t)$ with $t \in \mathcal{T}$ continuous, and for every $\delta>0$  we have
    $$ p_{\Pi}\Big\{\underset{t\in \mathcal{T}}{\sup}\;\; 
     || \Pi(t) - \Pi_0(t)|| < \delta \Big\} > 0.$$
\end{corollary}
\begin{proof}
    Given the logit link function $f(\cdot)$ is continuous and one-to-one, we have that for every $\delta>0$, there exists $\epsilon>0$ such that
\begin{equation}
\label{pos}
     \underset{t\in\mathcal{T}}{\sup} || f\{S(t)\}-f\{S_0(t)\}||_2 =\underset{t\in\mathcal{T}}{\sup} || \Pi(t)-\Pi_0(t)||_2 < \delta,
\end{equation}
    for all $S(t)$ having $\sup_{t\in\mathcal{T}} || S(t)-S_0(t)||_2 < \epsilon$, where $f\{S(t)\}$ indicates elementwise application of $f(\cdot)$ to the entries of $S(t)$. By Theorem \ref{thm:1} we know that $\sup_{t\in\mathcal{T}} || S(t)-S_0(t)||_2 < \epsilon$ has positive probability, which hence directly implies that also ${\sup}_{t\in\mathcal{T}} || \Pi(t)-\Pi_0(t)||_2 < \delta$ has non-zero probability, concluding the proof. 
\end{proof}

\section{Simulation studies}
\subsection{Simulations from the \texttt{NEX} model: additional details } \label{SM_SimulationNex}
In this section we provide additional details on the simulation study in Section \ref{SimulationNex}. In selecting the hyperparameters for the non-time factor matrices (mode 1 and mode 2) for the tensors $\mathcal{X}, \mathcal{Y}$, we utilize a heuristic which involves setting $K^*$, a user-defined estimate of the rank of the CP decompositions; this estimate should be based on the complexity of the dependence structure in exploratory data analysis, with a $K^* \leq 10$ appropriate for most applications. In this case, we suppose that our user-defined estimate of the dimension of the exemplar space is $K^* = 7$. 

We set the prior variance on the elements of the non-time factor matrices to be $\sigma^2 = 1/\sqrt{K^*}$ such that the elements of the tensors $\mathcal{X}, \mathcal{Y}$ are of order one a priori. Additionally, we utilize the following hyperparameters for the double-shrinkage prior \eqref{mgp:prior}-\eqref{mgp:prior:2}: $\alpha_1^K = \alpha_1^H=7$, $\beta_1^K=\beta_1^H=6$, $\alpha_2^K=2.5$, $\beta_2^K=7.5$, $\alpha_2^H=3$ and $\beta_2^H=12$. We set $\alpha_2^K / \beta_2^K > \alpha_2^H / \beta_2^H$ to favor a greater shrinkage in the latent dimension relative to the exemplar dimension, consistent with the fact that $H\leq K$. For the prior mean of each time-varying intercept, we utilize the empirical value of interaction prevalence: 
\begin{equation}
\label{eq:mu:0}
\mu_0=\log\left(\frac{n_{obs}/n_{tot}}{1-n_{obs}/n_{tot}}\right),
\end{equation}
where $\Omega_{ijt}$ is the Boolean tensor indexing the entries of our networks which are missing, $n_{obs}=\sum_{ijt} \Omega_{ijt} \mathcal{A}_{ijt}$ and $n_{tot}=\sum_{ijt} \Omega_{ijt}$.

We generate a single chain of $4000$ Hamiltonian Monte Carlo samples with \texttt{adapt delta} $= 0.95$ and \texttt{max treedepth} $= 12$, using the first half as warm-up. Markov Chain Monte Carlo diagnostics are satisfactory with good mixing of all entries of $\Pi$, and only a few sporadic divergences.

\begin{figure}[htp!]
\centering
	\begin{minipage}[t]{0.40\textwidth}
		\includegraphics[width=\textwidth]{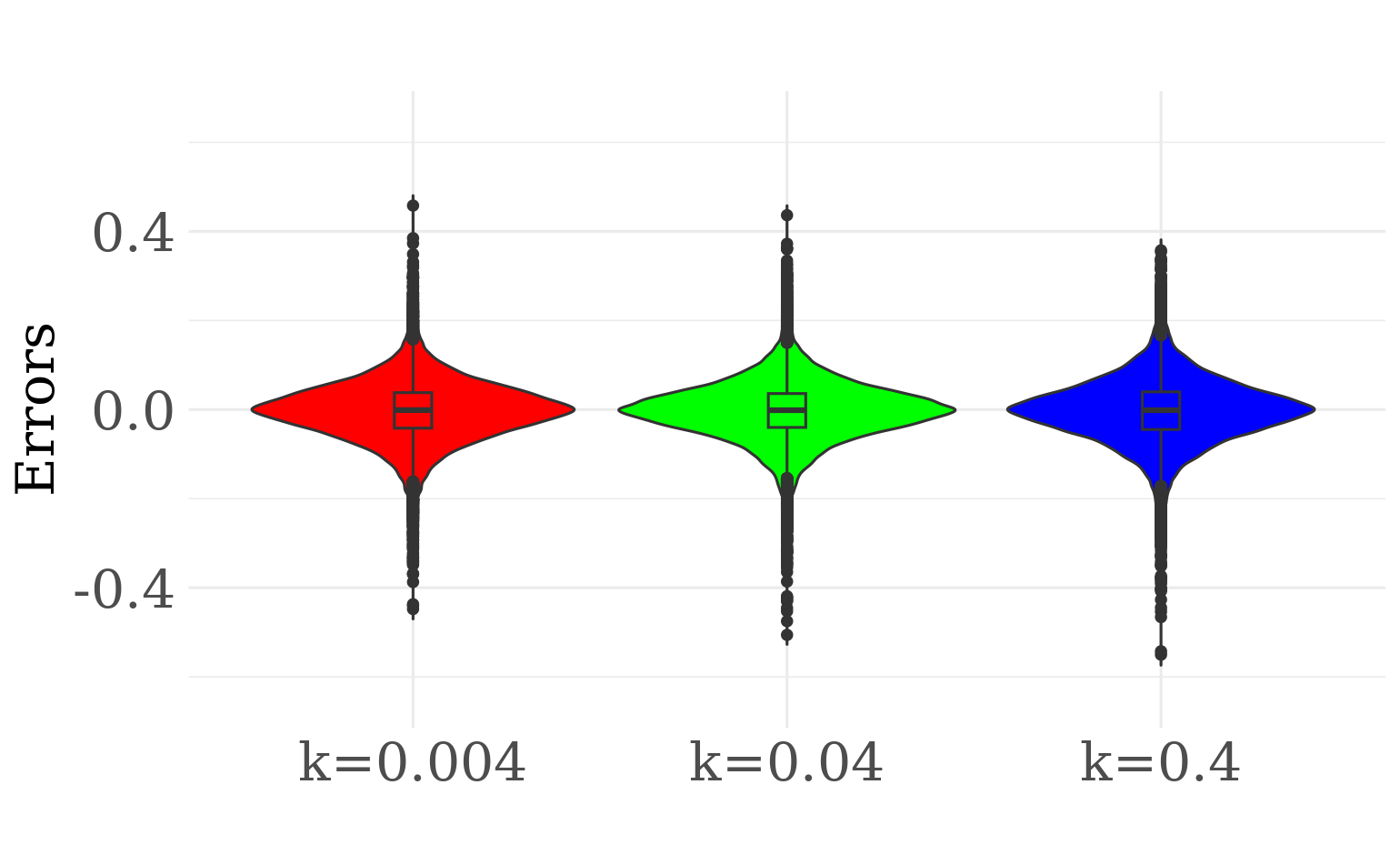}%
	\end{minipage}%
	\hspace{0.075\textwidth}
	\begin{minipage}[t]{0.40\textwidth}
		\includegraphics[width=\textwidth]{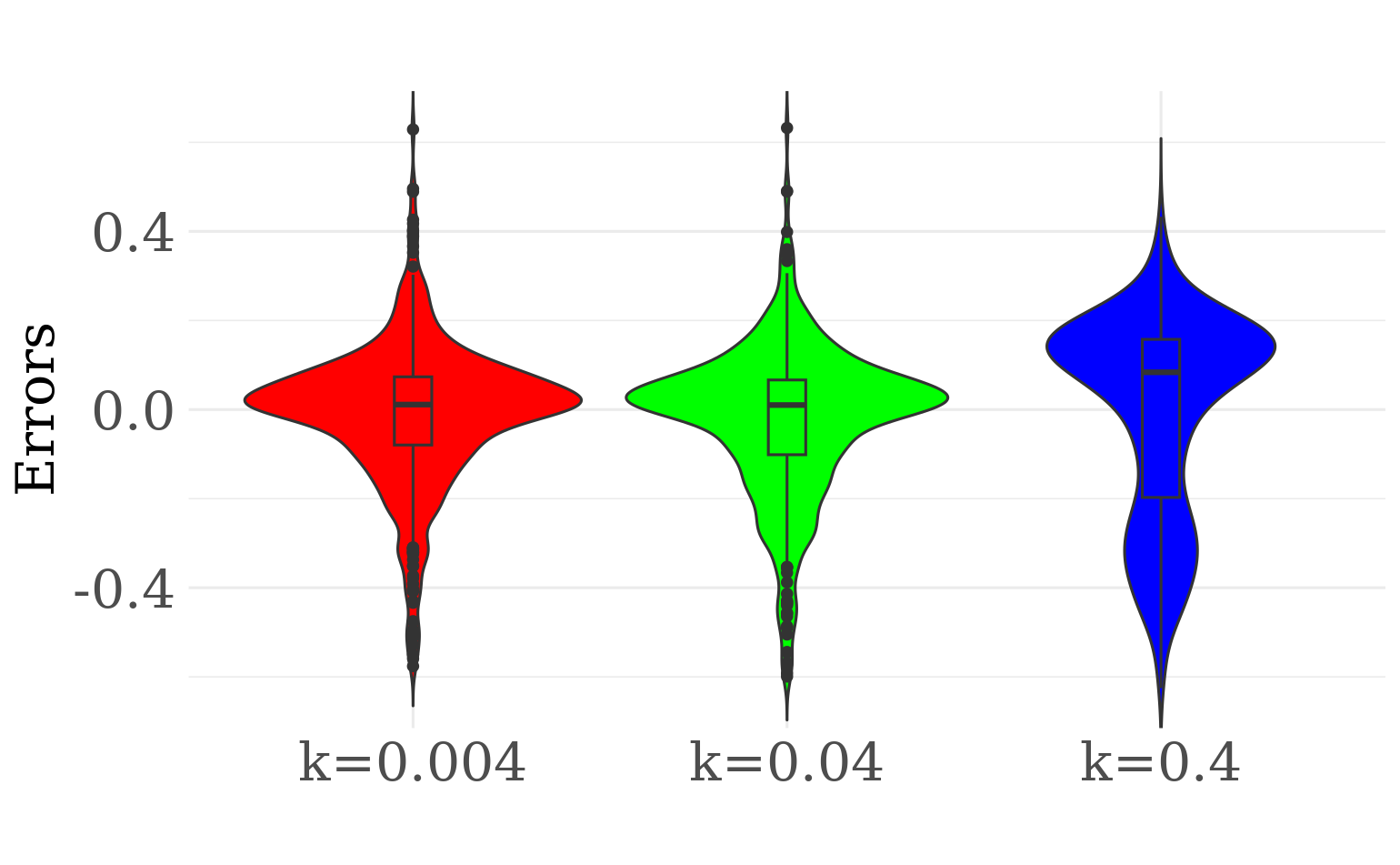}%
	\end{minipage}
	\caption{In sample (left) and out-of-sample (right) errors with respect to the true underlying probabilities when \texttt{NEX} is fit to synthetic data using length scales $k$ in $\{0.004, 0.04, 0.4\}$, respectively left, center and right of each subplot, with true length scale set to $0.1$ in the data generating process.}
\label{fig:sens}
\end{figure}

In general, the selection of the length scale parameter of the Gaussian processes $k_\mu, k_W$ should be guided by the strength of autocorrelation in the observed networks. To illustrate the robustness of the model to miss-specification in this parameter, we first fit \texttt{NEX} to the synthetic data with length scales $[0.004, 0.04, 0.4]$. As we can see from Figure \ref{fig:sens} the model is robust even to gross miss-specification. The model performs better when the length-scale is underestimated (too high covariance between time periods) rather than overestimated (too low covariance between time periods).  In the discussion that follows, all results refer to length scale values $k_{\mu}=k_{W}=0.02$ -- a moderate underestimate that still captures temporal dependence.

\subsection{Robustness to miss-specification} \label{SM_sim_misspecified}
In this subsection we assess the robustness of \texttt{NEX} to model miss-specification and evaluate it with respect to the higher-dimensional dynamic latent factor model of \cite{Durante.etal:2014}. We fit both \texttt{NEX} and the dynamic latent factor model to three synthetic data sets data generated from the \texttt{NEX} model, the dynamic latent factor model and a dynamic stochastic block model. Given that these data-generating mechanisms were formulated for the study of symmetric networks without self-edges, we restrict ourselves to this case in this subsection and correspondingly fit the symmetric variant of the \texttt{NEX} model specified in Section \ref{sec:model:formulation}. To evaluate robustness, we compare in-sample and out-of-sample performance of \texttt{NEX} and the dynamic latent factor model fitted to the data generated from the three models above. In each of the three cases, we consider networks consisting of $N=20$ nodes, observed at $T=40$ different time steps. 

Recall that the dynamic latent factor network model of \cite{Durante.etal:2014} imposes temporal dependence on the latent factors of an inner product network model: 
\begin{align}
\label{eq:DLF_sim_1}
p(\mathcal{A}_{ijt}=1| s_{ij}(t) &= \frac{1}{1 + e^{-s_{ij}(t)}} \\
\label{eq:DLF_sim_2} s_{ij}(t) &= \mu(t) + x_{i}(t)'x_j(t), \quad x_i, x_j \in \R^H \\
\label{eq:DLF_sim_3} x_{ih} &\sim \text{GP}(0, \tau_h^{-1}C_X) \quad h = 1, \ldots, H, \\
\label{eq:DLF_sim_4}(\mu_{t_1}, \ldots, \mu_{t_T})^\top &\sim \text{GP}(0, C_\mu)
\end{align}
where $c_X$  is a squared exponential kernel for the Gaussian process governing the distribution of any fixed component $h$ of the latent factor vector over time: $C_X(t,t') = \text{exp}(-\kappa_X | t - t'|)/2, C_\mu(t,t') = \text{exp}(-\kappa_\mu| t - t'|)/2$, and the effective dimension of the latent space is controlled by a multiplicative gamma shrinkage process on $\tau_h$. The above specification includes a slight modification to increase computational stability using an L1 exponential kernel rather than the L2 kernel originally specified. In our simulation, we set $H=2, \kappa_x = \kappa_{\mu} = 0.04$, and all other hyperparameters according to the recommendations in \cite{Durante.etal:2014}. Because this model is extremely flexible compared to the more parsimonious model proposed, it provides a good test of the ability of \texttt{NEX} to recover complex data structures while still maintaining a relatively low dimension on the exemplar space, $K$.

Stochastic block models describe 
networks via community structure, and are typified by the assumption of stochastic equivalence: the nodes of a network can be partitioned into groups such that the probability of connection between any pair of nodes depends only on the block assignment of the respective nodes. The variant we consider permits a dynamic structure by assuming that block assignment is constant over time, but that within- and between-block connection probabilities vary over time. Formally, we suppose that each of the nodes $i = 1, \ldots, n$ belongs to one of the $B<n$ blocks with block membership indicated by the vector $Z_i \in \R^B$ having a one in the position corresponding to the group of node $i$ with all other elements zero.  The probabilities of connection within and between blocks at time $t$ are described by the $B \times B$ symmetric block connection matrix $\Pi(t)$. This dynamic stochastic block model (SBM) has the following formulation: 
\begin{gather}
\label{eq:dyn:1} pr\{\mathcal{A}_{ijt}=1| Z_i, Z_j, \Pi(t)\} =  Z_i^\top \Pi(t) Z_j \\
\label{eq:dyn:2} \hbox{logit}^{-1}\left\{\Pi_{b,b'}(.)\right\} \sim \mbox{GP}(\mu_{b, b'}, \Sigma_B), 
\end{gather} 
for $i=1,\ldots,n, j=1,\ldots,n$, and
$\Sigma_B(t,t')=\text{exp}(-\kappa_B | t - t'|)/4$ with $\kappa_B=0.04$. 
To summarize: the membership of a node is fixed across time, while the probabilities of observing a edge between different groups evolves. This provides a useful test to see if \texttt{NEX} is capable of effectively modeling the block structure if it is indeed present in the data. 

The data generated from the symmetric \texttt{NEX} variant is simulated with $H_{true}=2, K_{true}=5$, $C_x(t,t') = C_{\mu}(t,t') = \text{exp}(-\kappa |t-t'|)$, $\kappa=0.04$, $\sigma^2 =2/\sqrt{K_{true}}$, $\mu_0 = 0$, and with all hyperparameters equal to their values above.

\begin{figure}[htp]
	\begin{minipage}[t]{0.99\textwidth}
		\includegraphics[width=\textwidth]{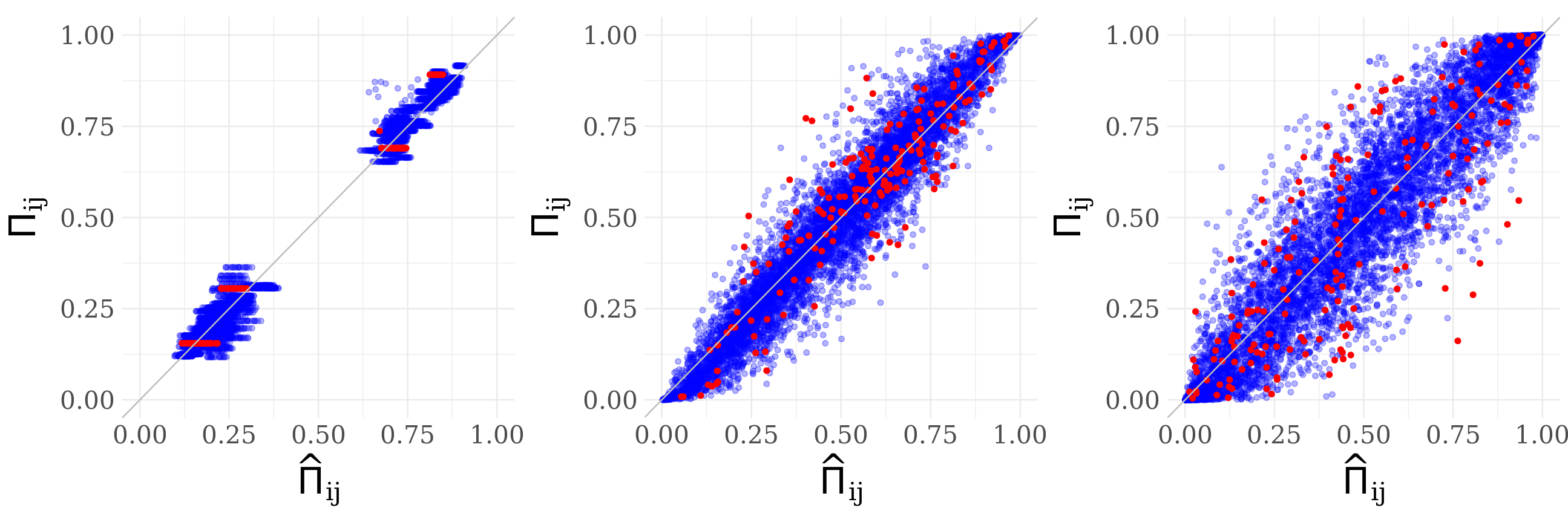}%
	\end{minipage}%
	\hspace{0.0\textwidth}
	\begin{minipage}[t]{0.99\textwidth}
		\includegraphics[width=\textwidth]{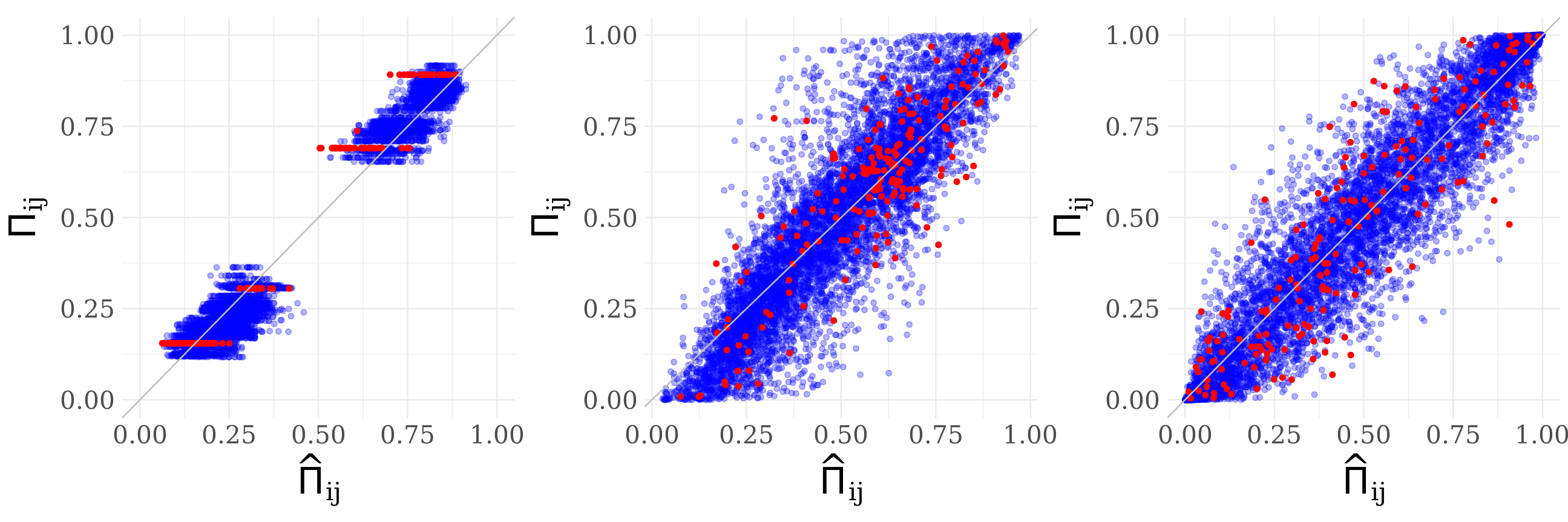}%
	\end{minipage}
    \caption{In-sample (blue) and out-of-sample (red) performance of \texttt{NEX} (top row) and dynamic latent factor model (bottom row) for synthetic data generated from the dynamic SBM \eqref{eq:dyn:1}-\eqref{eq:dyn:2} (left column), \texttt{NEX} (center column), and the dynamic latent factor model \eqref{eq:DLF_sim_1}-\eqref{eq:DLF_sim_4} (right column).}
    \label{fig:miss}
\end{figure}

To each of the above synthetic datasets we fit the symmetric \texttt{NEX} model and the dynamic latent factor model. In fitting \texttt{NEX}, we set the upper bound on the dimension of the latent space (H) and exemplar space (K) as  $H=K=10$ for data simulated from the dynamic SBM and \texttt{NEX}, while we set $H=10, K=20$ for the data simulated from the dynamic latent factor model given its high flexibility; in fitting the dynamic latent factor model, we set the upper bound on the dimension of the latent space (H) as H=5 for all synthetic data sets due to its greater flexibility and tendency towards overfitting. All \texttt{NEX} hyperparameters are set as above in Section \ref{SimulationNex} with the exception of the GP covariance, which we set to have a moderate level of temporal dependence using $C_x(t,t') = C_{\mu}(t,t') = \text{exp}(-\kappa|t-t'|)$ with $\kappa= \kappa_x = \kappa_\mu = 0.02$ consistent with the procedure in Section \ref{SimulationNex}; the same GP covariance hyperparameters are used for the dynamic latent factor model. All other hyperparameters values are set to the values used in \cite{Durante.etal:2014}, detailed in Section \ref{SM_app_priors}. 

For each model fit and each synthetic data set, we generate a single chain of $4000$ samples, discarding the first $1000$ as burn-in. 
We then assess both the in-sample and out-of-sample recovery of the true underlying probabilities. From Figure \ref{fig:miss} and Figure \ref{fig:miss:2}, it is clear that the dynamic latent factor model shows a greater degree of overfitting, especially for data generated from the symmetric \texttt{NEX} model and the dynamic stochastic block model, while providing a slightly better fit to data generated from the dynamic latent factor model. Hence, we find that \texttt{NEX} outperforms the dynamic latent factor model in terms of out-of-sample error for data generated from \texttt{NEX} and the dynamic SBM, while performing slightly worse than the dynamic latent factor model for data generated from that model itself, as expected.

\begin{figure}[htp]
\centering
	\begin{minipage}[t]{0.49\textwidth}
		\includegraphics[width=\textwidth]{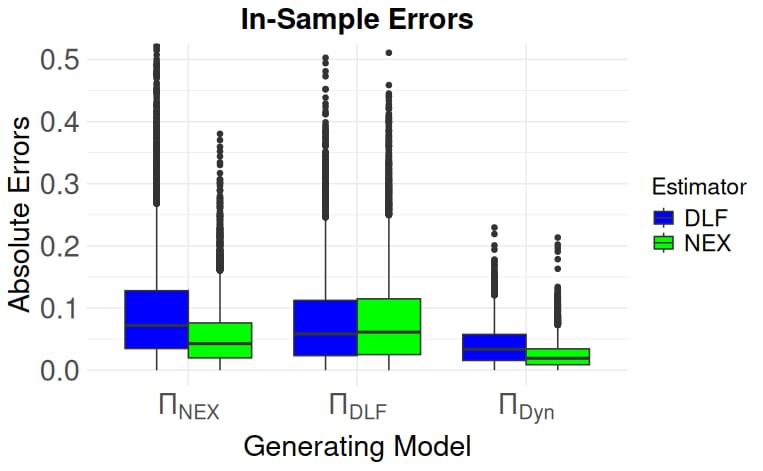}%
	\end{minipage}%
	\hspace{0.0\textwidth}
	\begin{minipage}[t]{0.49\textwidth}
		\includegraphics[width=\textwidth]{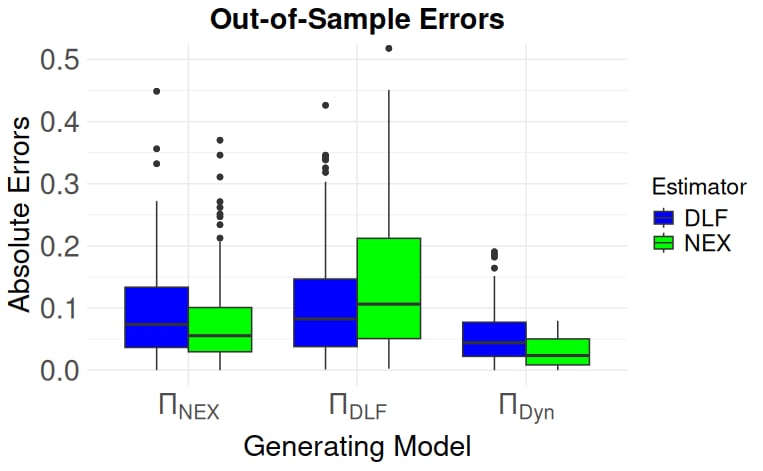}%
	\end{minipage}
	\caption{In-sample (left) and out-of-sample(right) performance of \texttt{NEX} and the dynamic latent factor model  (\texttt{DLF}) for synthetic data generated from \texttt{NEX}, the dynamic latent factor model \eqref{eq:DLF_sim_1}-\eqref{eq:DLF_sim_4}, and the dynamic stochastic block model (Dyn) \eqref{eq:dyn:1}-\eqref{eq:dyn:2}.}
\label{fig:miss:2}
\end{figure}


\section{Zackenberg Application Details}
\label{SM_app}

\subsection{Data preparation}\label{SM_app_taxa}
Data preparation included taxonomical re-assignment to eliminate ambiguity, construction of a seasonally aligned time measure, construction of occurrence matrices and co-occurrence arrays, and subsequent subsetting. Although all data analysis was performed with taxonomical labels, the figures below omit taxa names or utilize generic substitutes due to data sharing limitations; where possible, limited results relating to specific taxa of interest are presented. 

To provide an overview of the data, we construct a meta-network by summing binary interaction networks over all years, encompassing 49 weeks of data. Of 3315 flower-pollinator pairs, only 594 are actually observed in an interaction. Of these, approximately 50\% of the interactions occur in only one of the 49 observed weeks. Only 35 interactions are observed in at least ten separate weeks, and the most common interaction is observed in 31 separate weeks. The meta-network shown in Figure \ref{fig:sm_metanetwork} shows the sparsity of interactions in this setting. 

\begin{figure}[htp]
    \centering
    \includegraphics[width=0.75\linewidth]{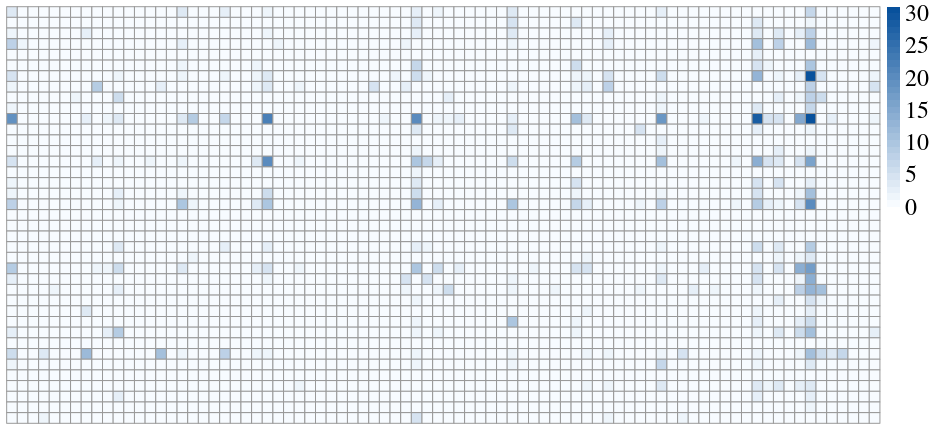}
    \caption{Meta-network of insect-plant interactions: rows indicate plant taxa, columns indicate insect taxa, and cell color indicates the number of distinct weeks an interaction is observed.}
    \label{fig:sm_metanetwork}
\end{figure}

The data consists of 39 plant taxa, of which all are resolved to the species level except for the genus of the taxonomically challenging genus of whitlow grasses \textit{Draba}, which is not otherwise represented in the data and therefore presents no ambiguity. The insect taxa present more of a challenge: the raw data consists of 114 insect taxa including 42 incompletely identified taxa: 36 identified only to the genus level and the remaining 6 identified only to the family level. 

In the case where a recorded genus has at least one species resolved but also has individuals identified only at the genus level, we pool all observations into a single operational taxon and similarly for cases where a recorded genus has no species resolved. Although these operational taxa may in fact consist of different species with different population sizes, poorly identified taxa tend to be ecologically similar with respect to the pollination niche, a general phenomenon called niche conservatism \citep{Wiens}. In contrast, when individuals are identified to the family level only, we discard the observations due to ambiguity between these individuals and individuals from fully resolved species in the same family; one exception to this procedure is made for the family Scathophagidae, which we treat as an operational taxon. 

After executing the taxa correction protocols above, we obtain a collection of networks spanning 82 operational taxa: 36 consisting of genus-level groupings, 1 corresponding to the family Scathophagidae as described above, and the remainder corresponding to fully resolved species. Each taxa is observed in at least one pollination interaction, with a median of 8 and a maximum of 717 interactions observed, including repeated observations of the same plant-pollinator pair in a single week. Most species are very infrequent pollinators. 


Since the timing of snow melt varies from year to year in the study region, we align the nominal week indices by defining time via cumulative heat elapsed since snow melt, or growth degree days. We compute the cumulative sum of growth degree-days (GDD) in each week index as:
$$GDD = \max(0,(T_{MAX} + T_{MIN})/2 - T_{BASE}),$$
where $T_{MAX}, T_{MIN}$, and $T_{BASE}$ refer to the maximum, minimum and base air temperature, using $0^\circ$ C as the base temperature. The temperature is recorded as the air temperature 200 cm above the ground from the Greenland Ecosystem Monitoring program \citep{climatebasis}.


Plant occurrences are assumed to be observed without error, while insect occurrences are estimated using three different procedures: raw occurrence data - an insect is present only if it is observed in some interaction in that week, padded occurrences 1 - an insect is present if it was observed in an interaction in week $t$ or in either of week $t-1$ or $t+1$, padded occurrences 2 - an insect is present in week $t$ if it is present in week $t-1$ and $t+1$. A visualization of the three occurrence constructions is provided in Figure \ref{fig:sm_occurrences}. For the above subset of common insects and plants, we increase the number of insects assumed to be present by 56 to 81\% in scenario 1 and by 4.8 to 12\% under padding scenario 2.

\begin{figure}[htp]
    \centering
    \includegraphics[width=0.65\linewidth]{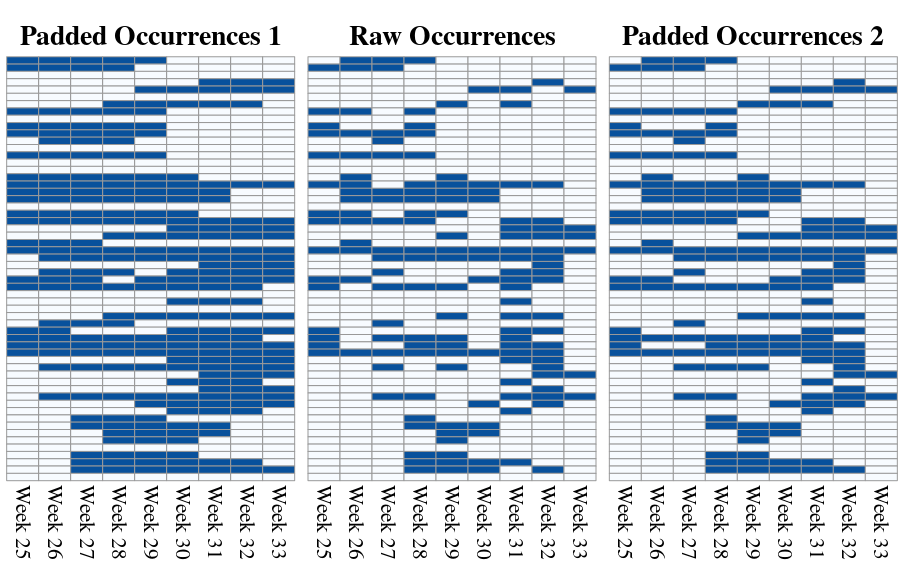}
    \caption{Estimated insect species occurrences for 2011 under three different scenarios, rows indicate insect taxa and a darkened cell indicates the presence of the taxa in that week.}
    \label{fig:sm_occurrences}
\end{figure}

After preparing the supplementary data as previously described, we subset the raw networks to include only plants and insects observed in an interaction over at least two weeks and remove two weekly networks in which only one plant was present. Outlier networks in 2011 result from observations all coming from a single plant species in those weeks - \textit{Arenaria pseudofrigida} in both cases. The above subsetting results in a subnetwork of the 35 most common plants and 58 most common insects for 47 weeks.

\subsection{Full model specification}
\label{SM_app_priors}
 For \texttt{NEX} we set $\alpha_1^K = \alpha_1^H=7$, $\beta_1^K=\beta_1^H=6$, $\alpha_2^K=2.5$, $\beta_2^K=7.5$, $\alpha_2^H=3$ and $\beta_2^H=12$ for the double shrinkage prior \eqref{mgp:prior}-\eqref{mgp:prior:2}, and empirically compute $\mu_0$ through \eqref{eq:mu:0}.
 We define the temporal dependence of the latent space via the separable covariance $\Sigma = \exp(-k_1 D_1) \exp(-k_2 D_2)$ where $D_1(w,w') = |w - w'|$ and $D_2(r,r') = |r - r'|$. We set the temporal dependence for the seasonal intercept $\mu_w$ with $\Sigma_\mu = \exp(-k_1 D_1)$.  For both \texttt{NEX} and the dynamic latent factor model, we utilize $k_1=k_{\mu}=0.0075$ and $k_2 = 0.075$, resulting in moderate prior auto-correlation. 
 
 For the dynamic latent factor model, we specify the following: 
\begin{align*}
    E[a_{ijwr} \mid \pi_{ijwr} ] &= \pi_{ijwr} = \frac{1}{1 + e^{-s_{ijwr}}} \\
    s_{ijwr} &= \mu_w + \alpha_i + \gamma_r + \beta c_{wr} + x_{iwr}^\top y_{jwr} \\
    x_{ih}(\cdot) & \sim GP(0, \tau_h^{-1} \Sigma) \\
    y_{ih}(\cdot) & \sim GP(0, \tau_h^{-1} \Sigma) \\
    \mu(\cdot) & \sim GP(0, \Sigma_\mu) \\
    \alpha_i & \sim N(0, \sigma^2_\alpha) \text{ for }i = 1, \cdots, M \\
    \gamma_r & \sim N(0, \sigma^2_\gamma) \text{ for }r = 1, \cdots, R \\
    1/\sigma_\alpha^2, 1/\sigma_\gamma^2 &\sim  \text{Ga}(\eta_0/2, \eta_0 \sigma^2_0/2) 
\end{align*}
where $\Sigma$ and $\Sigma_\mu$ are defined as above for \texttt{NEX}. We set $\eta_0 = 100$, $\mu_{\beta}=0$, $\sigma_{\beta}=1$. The shrinkage parameter $\tau_h^{-1}$ is defined as in \cite{Durante.etal:2014}: 
$$\tau_h = \prod_{k = 1}^h \vartheta_k, \vartheta_1 \sim \text{Ga}(a_1, 1), \vartheta_k \sim \text{Ga}(a_2, 1), k= 2, \cdots, H,$$
with $a_1 = a_2 = 2$.
\subsection{Sensitivity analysis with respect to occurrence construction}
\label{SM_app_sens}

In this section, we explore the sensitivity of the results to the occurrence specification defined in Section \ref{SM_app_taxa} above. Although exploratory analysis of raw data suggests a positive effect of temperature on interaction richness, we find the effect disappears when conditioning on co-occurrence with either of the three constructions proposed and fitting the \texttt{NEX} model, as can be seen in Figure \ref{fig:sm_temperature}. Similarly, the raw data show a seasonal trend in interaction richness, with pollination activity peaking in mid-summer. However, when controlling for co-occurrence in the \texttt{NEX} model, we find a weak reversed trend; see Figure \ref{fig:sm_seasonality}. 

\begin{figure}[htp]
	\begin{minipage}[t]{0.45\textwidth}
		\includegraphics[width=\textwidth]{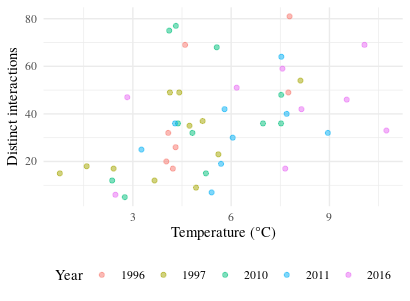}%
	\end{minipage}%
	\hspace{0.0\textwidth}
	\begin{minipage}[t]{0.45\textwidth}
		\includegraphics[width=\textwidth]{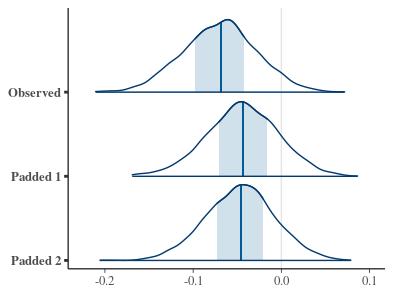}%
	\end{minipage}
	\caption{Observed relationship between temperature and pollination activity (left) versus estimated effect in the conditional \texttt{NEX} model (right) using observed occurrence, padded occurrence scenario 1, and padded occurrence scenario 2 with posterior mean and 50\% credible interval.}
    \label{fig:sm_temperature}
\end{figure}

\begin{figure}[htp]
	\begin{minipage}[t]{0.45\textwidth}
		\includegraphics[width=\textwidth]{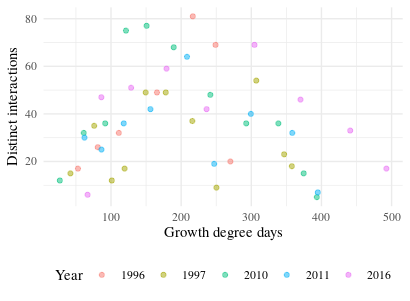}%
	\end{minipage}%
	\hspace{0.0\textwidth}
	\begin{minipage}[t]{0.49\textwidth}
		\includegraphics[width=\textwidth]{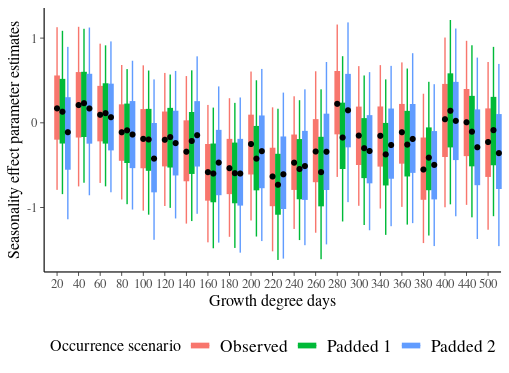}%
	\end{minipage}
	\caption{Observed relationship between seasonality and pollination activity richness (left) versus estimated effect in the conditional \texttt{NEX} model (right) using observed occurrence, padded occurrence scenario 1, and padded occurrence scenario 2 with posterior means, 50\% and 90\% credible intervals.}
    \label{fig:sm_seasonality}
\end{figure}

Sensitivity analysis for the plant random effects accommodating degree heterogeneity is provided in Figure \ref{fig:sm_plant_effects}. The results are broadly consistent across the three scenarios of occurrence, with all the model specifications suggesting the importance of Plant 11 (\textit{Dryas octopetala}) and Plant 27 (\textit{Saxifraga hyperborea}) for the arctic ecosystem. 

\begin{figure}[htp]
    \centering
    \includegraphics[width=0.85\linewidth]{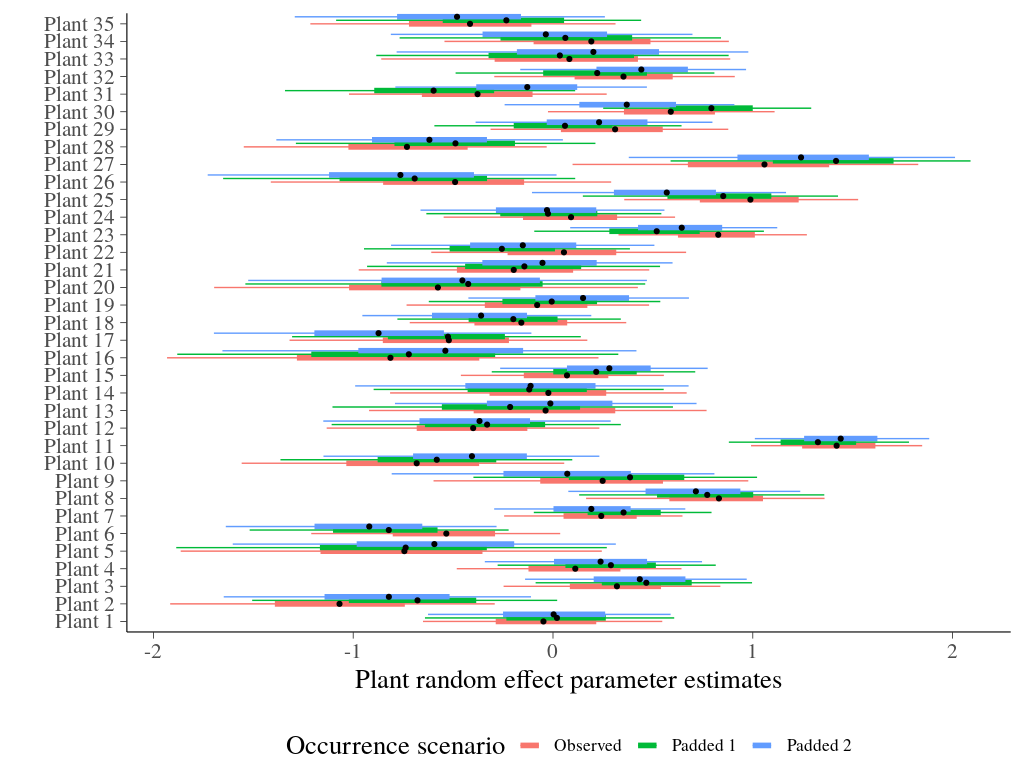}
    \caption{Posterior mean, 50\% and 90\% credible intervals for plant random effects obtained from the \texttt{NEX} model fit under three separate occurrence scenarios.}
    \label{fig:sm_plant_effects}
\end{figure}

Sensitivity analysis with respect to the year random effects shown in Figure \ref{fig:sm_year_effects} suggests the broad stability of the arctic pollination network over the two decades studied, with a slight increase in baseline pollination activity among co-occurring pairs in recent years. The cause of this uptick is uncertain, with one potential explanation being a reduction in alternatives increasing the propensity for interaction between any given pair consisting of a present insect and a present flower. 

\begin{figure}[htp]
    \centering
    \includegraphics[width=0.65\linewidth]{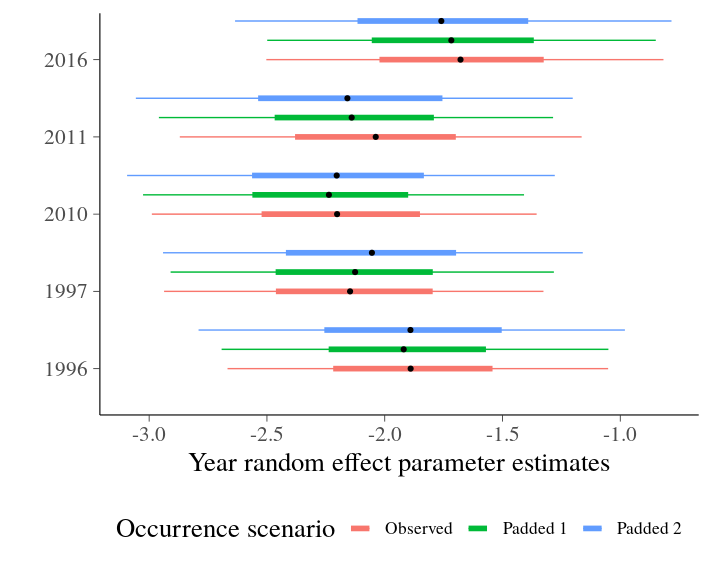}
    \caption{Posterior mean, 50\% and 90\% credible intervals for year random effects obtained from the \texttt{NEX} model fit under three separate occurrence scenarios.}
    \label{fig:sm_year_effects}
\end{figure}

Figure \ref{fig:sm_traceplots} shows good mixing and convergence of the log likelihood to the stationary distribution under all three occurrence scenarios. 

\begin{figure}[htp]
    \centering
    \includegraphics[width=0.85\linewidth]{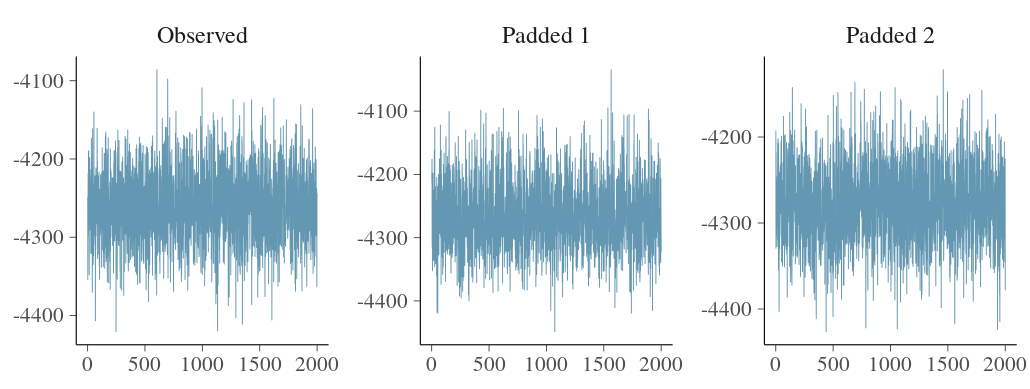}
    \caption{Trace plots of the log likelihood obtained from the \texttt{NEX} model fit under three separate occurrence scenarios indicate good mixing and convergence of the respective chains to their stationary distributions.}
    \label{fig:sm_traceplots}
\end{figure}

\subsection{New high posterior probability interactions predicted by \texttt{NEX}}\label{sm_app_new_interactions}

Because we are interested in the impacts of shifting phenologies on pollination, it is important to investigate which currently unobservable interactions might occur in the future given overlap of plant blooming and insect foraging. Focusing only on rare insects, we predict \textit{Saxifraga hyperborea} to be a potential flower resource for: \textit{Euxoa adumbrata, Agriades glandon, Entephria kidluitata, Praon brevistigma, Propanteles fulvipes, Neolaria prominens, Orthocladius}. Additionally \textit{Salix acrtica} is a predicted flower resource for the genus \textit{Atractodes}. These results are summarized in Figure \ref{fig:app_future_interactions}.

\begin{figure}[htp]
    \centering
    \includegraphics[width=0.9\linewidth]{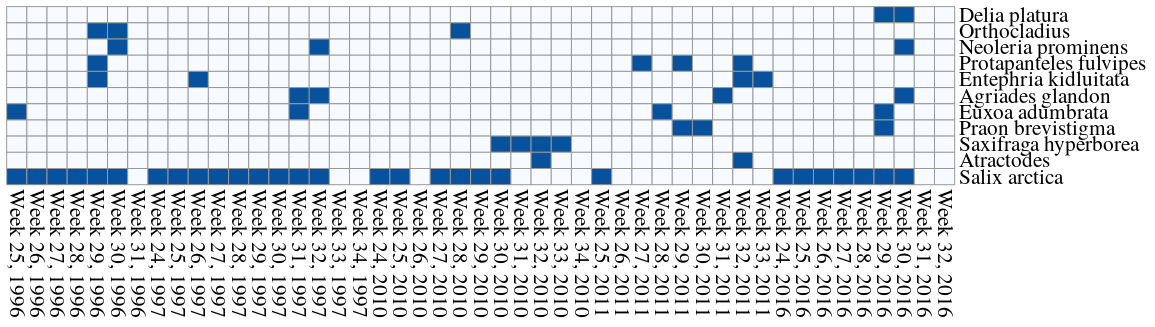}
    \caption{Occurrence patterns for selected rare insects and the plants \textit{Saxifraga hyperborea} and \textit{Salix arctica} with which they have a mean posterior interaction probability greater than 0.5. Occurrence information is generated using Padded Occurrence Scenario 1.}
    \label{fig:app_future_interactions}
\end{figure}

A basic question of interest is which interactions have been missed due to undersampling or difficulty of observation. We aggregate interactions by insect family and plant species by taking the maximum of posterior mean interaction probabilities across insects within that family and summarize these results in Figure \ref{fig:missing_future_interactions}.  We identify many previously unobserved high posterior probability interactions for the families Agromyzidae (insect family 1) and Scathophagidae (insect family 23) - and Braconidae (insect family 4) and Ichneumonidae (insect family 13).

\begin{figure}[htp]
    \centering
     \includegraphics[width = 0.75 \textwidth]{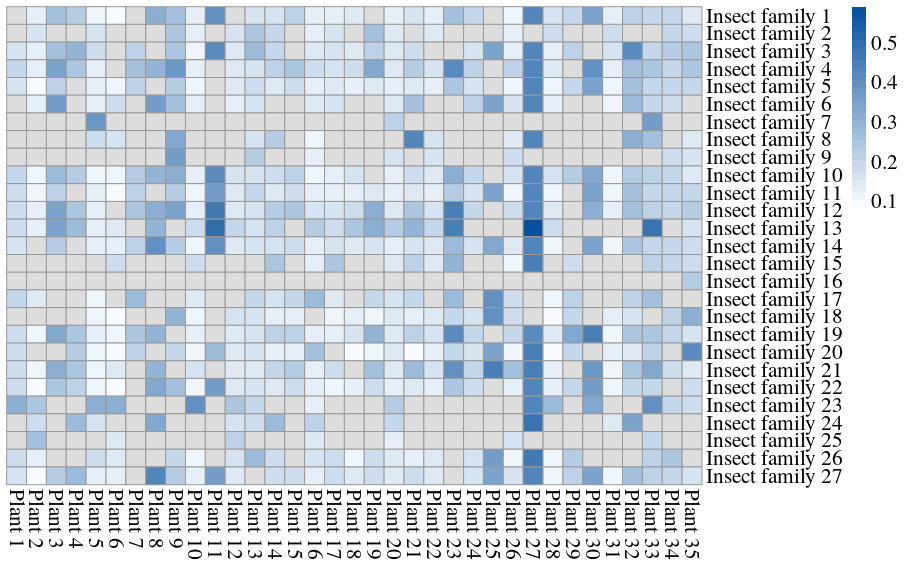}
    \caption{Maximum of posterior mean interaction probabilities across all weeks, by insect family; model fit using Padded Occurrence Scenario 1. Grey squares indicate observed interactions in the meta-network.}
    \label{fig:missing_future_interactions}
\end{figure}

\end{appendix}
\end{document}